\documentclass[%
 reprint,
superscriptaddress,
 amsmath,amssymb,
 aps,
pra,
]{revtex4-2}

\usepackage{graphicx}
\usepackage{subcaption}
\usepackage{dcolumn}
\usepackage{bm}
\usepackage[colorlinks=true,linkcolor=blue,citecolor=blue]{hyperref}

\usepackage{amsthm}
\usepackage{braket}
\usepackage{dsfont}
\usepackage[dvipsnames,table,xcdraw]{xcolor}
\usepackage{tikz}
\usetikzlibrary{positioning}

\newtheorem{theorem}{Theorem}
\newtheorem{algorithm}{Algorithm}
\newtheorem{coro}{Corollary}
\newtheorem{conjecture}{Conjecture}

\newcommand{\Tr}{\text{Tr}}

\newcommand{\chn}{\mathcal{E}}
\newcommand{\chnd}{\mathcal{F}} 
\newcommand{\chnds}{\mathcal{F_{\mathsf{gen}}}} 
\newcommand{\chndh}{\mathcal{F_{\mathsf{Herm}}}} 

\newcommand{\poi}{\alpha} 
\renewcommand{\pof}{\beta} 

\newcommand{\rf}{\gamma} 
\newcommand{\rfd}{\Gamma} 

\newcommand{\rfsb}{\xi} 
\newcommand{\rfdsb}{\Xi} 

\newcommand{\tp}{\textsf{T}}
\newcommand{\ret}[2]{\hat{#1}_{#2}}
\newcommand{\retd}{\ret{\chn}{\rf}} 
\newcommand{\ir}[1]{\frac{1}{\sqrt{#1}}}
\newcommand{\depolarise}{\mathcal{D}}
\newcommand{\dephase}{\mathcal{Z}}

\DeclareMathOperator*{\argmin}{\arg\!\min}

\hypersetup{
    colorlinks=true,
    linkcolor=blue,
    filecolor=magenta,      
    urlcolor=blue,
    pdftitle={Overleaf Example},
    pdfpagemode=FullScreen,
    }

\begin{document}

\preprint{APS/123-QED}

\title{Schr\"odinger Bridges via the Hacking of Bayesian Priors \\ in Classical and Quantum Regimes}

\author{Clive Cenxin Aw}%
\email{ccx.aw@nus.edu.sg}
\affiliation{Centre for Quantum Technologies, National University of Singapore, 3 Science Drive 2, Singapore 117543, Singapore}

\author{Peter Sidajaya}%
\email{p.sidajaya@u.nus.edu}
\affiliation{Centre for Quantum Technologies, National University of Singapore, 3 Science Drive 2, Singapore 117543, Singapore}


\date{\today}

\begin{abstract}
Bayes’ rule is widely regarded as the canonical prescription for belief updating. We show, however, that one can arbitrarily preserve pre-specified beliefs while appearing to perform Bayesian updates via ``prior hacking'': engineering a reference prior distribution such that, for a fixed channel and evidence, the update matches a chosen target distribution. We prove that this is generically possible in both classical and quantum settings whenever Bayesian inversions are well-defined (with the Petz recovery map as the quantum analogue to Bayes' rule), and provide constructive algorithms for doing so. We further establish a duality between prior hacking and Schrödinger bridge problems (a key object in statistical physics with applications in generative modelling), yielding in the quantum setting a unique, inference-consistent selection among candidate bridges. This formally establishes the Bayes-like updating that Schrödinger bridges are performing with respect to the process as opposed to the reference prior, both in classical and quantum settings.
\end{abstract}

\maketitle
\section{Introduction}
It is difficult to imagine a clearer departure from rational inference than an agent who never alters their beliefs when presented with new data. It strikes us as epistemically pathological \cite{Battaly2020-dogma-philo, cassam2016vice-dogma-philo}, both in individual reasoning agents and more sociological contexts \cite{savion2009clinging-dogma-social,siebert2023effective-dogma-social,anderson1980perseverance-dogma-social}. This is why in robust accounts of inference, the responsiveness to evidence is a minimal requirement \cite{jacobs-changing-mind, chan2005revision-BK, jeffrey1990logic-BK}. 

Bayes’ rule, alongside its natural extensions to soft evidence \cite{chan2005revision-BK, jeffrey1990logic-BK,jaynes2003}, is widely taken to be \textit{the} canonical prescription for such rational responsiveness. It formalizes how a prior distribution should be transformed in the presence of new data accounting not only for relevant conditional probabilities but also one's \textit{reference prior} (which encodes the background information or, more generally, some initial best guess or belief). For these reasons and more, Bayes' rule is regarded as a baseline, if not a foundation, to robust updates of one's information. Its use is ubiquitous, underpinning social, statistical, and physical sciences while finding applications in machine learning and data engineering tasks \cite{feller1968introduction,jessop2018-bayes-app,chivers2024-bayes-app,tipping2003bayesian-bayes-app}. 

The work at hand examines an obscure cross-section between these two vastly different approaches to belief kinematics \cite{jeffrey1990logic-BK,zhou2014belief-BK,spohn1988ordinal-BK,van1980rational-BK,bradley2007kinematics-BK}. We ask under what circumstances is it possible to preserve one’s beliefs exactly while doing so under the guise of performing genuine Bayesian updates. We show that this can be achieved through \emph{prior hacking}: engineering a reference prior, defined for relevant transition probabilities (i.e. channels) and observed evidence, such that the ``update'' coincides with a pre-specified distribution. In this sense, one appears to update rationally via Bayes' rule while in fact remaining doxastically stationary; doing what is most pathological while appearing to do what is most sound, so to speak. 

As we will see in this work, it turns out prior hacking is generically possible in the classical setting, aside from a measure-zero set of channels. Specifically, it is \textit{always possible} to prior hack for a given channel, evidence and fixed conclusion whenever Bayesian update is well-defined (Theorem \ref{thm:classical-prior-hacking-channel-forall}). We provide a constructive algorithm (Algorithm \ref{algo:classical}) to obtain these hacking priors, when they exist. These are developed in Section \ref{sec:classical}.

This leads us to our central observation, which connects this question about belief kinematics to notions and applications in transport physics. We show how prior hacking is categorically parallel to the construction of Schr\"odinger bridges \cite{fortet1940resolution-SB-Org,leonard2013survey-SB-survey,chen2016relation-SB-Intro,chetrite2021schrodinger-SB-Org,orland2025schrodinger-SB-gen,georgiou2015positive} (Theorem \ref{thm:CSB-is-CPH}), a key object in statistical physics with growing applications in generative modelling and other computational tasks \cite{de2021diffusion-SB-ML,xie2024bridging-SB-ML,heng2024diffusion-SB-ML}. Specifically, we show that Schr\"odinger bridges can be seen as a dual to prior hacking---a kind of \emph{process hacking}. These are detailed in Section \ref{ssec:classical-SB}.

Section \ref{sec:quantum} extends our analysis to the quantum regime, focusing on a proposed quantum analogue to Bayes’ rule, the Petz Recovery map \cite{petz,petz1,barnum-knill,petzisking2022axioms, liu2026-extendedpetz,aw2026thesis}. We derive sufficient conditions under which prior hacking is possible for the Petz map (Theorem \ref{thm:quantum-surjective}, Algorithm \ref{algo:quantum}), covering the vast majority of cases. Section \ref{ssec:quantum-SB} elaborates on how the these constructions match algorithms for quantum Schr\"odinger bridges. Most notably, we show that correspondence to quantum prior hacking can help pick out unique quantum Schr\"odinger bridges from the existing continuum of candidates that exists in the current literature (Theorem \ref{thm:QSB-is-QPH}).  

We also complement the formal results with examples and geometric intuitions in both classical (Section \ref{sec:c-example-text}) and quantum settings (Section \ref{ssec:quantum-examples}). Section \ref{sec:concl} concludes with remarks on the foundational implications of these findings and directions for further investigation.

\begin{center}
    \begin{figure*}[ht]
    \centering
    \includegraphics[width=0.9\textwidth]{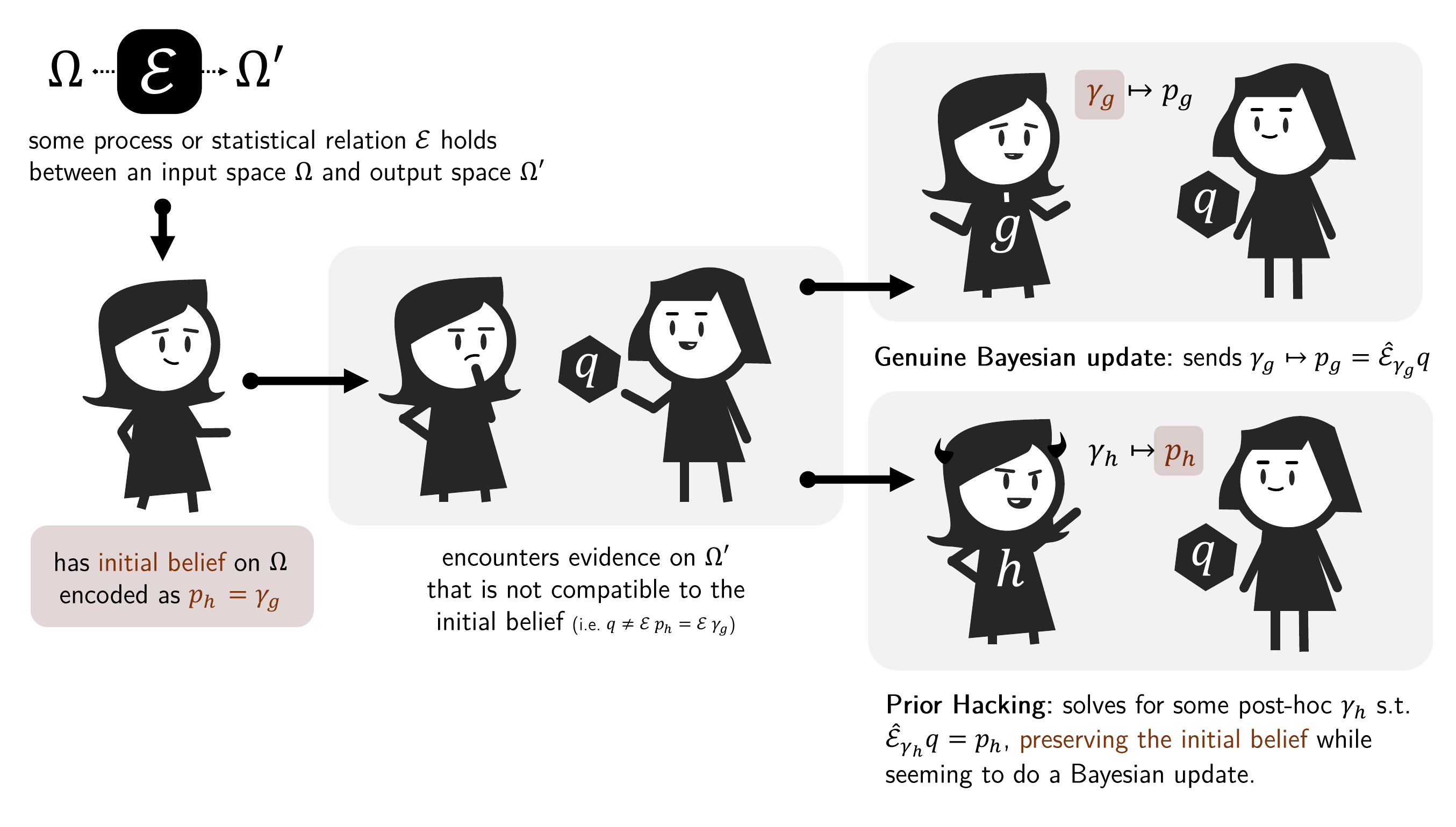}
    \caption{A cartoon illustration of prior hacking, which is detailed in Section \ref{sec:classical-prior-hacking-introduce}.}
    \label{fig:prior-hacking-illus}
\end{figure*}
\end{center} 

\section{Classical Setting} \label{sec:classical}
\subsection{The Classical Formalism} \label{ssec:classical-setting}
Before introducing what we mean by prior hacking, we delineate the formalism used---namely, that of stochastic matrices and Bayesian updating. We work with a discrete probability space $\Omega$ with $d$ elements where probability distributions $p:\Omega \longrightarrow [0,1]$ live. These can be represented by column probability vectors, which we also denote as $p$ to avoid notational encumbrance \footnote{While we have overloaded the symbol $p$ by assigning it to a function and a vector, in equations, the function $p(\cdot)$ will always appear with brackets, while the vector $p$ would not have any. This also applies to the other objects that we are going to introduce next.}. These exist in $\mathds{R}^d$ with entries $p(x)$, such that $ \forall p:\sum_{x} p(x)=1$ and $\forall x :p(x) \geq 0$. The set of all probability vectors in $d$ dimensions form a probability simplex $\Delta^{d-1}$ in $d-1$ dimensions.

Consider two probability spaces $\Omega$ and $\Omega'$ with dimensions $d$ and $d'$ and simplexes $\Delta^{d-1}$ and $\Delta^{d'-1}$. A stochastic process $\chn(\cdot|x)$ is a conditional probability function that assigns a probability to $\Omega'$ given that $x\in\Omega$ has been observed. We can transform a probability function $p$ in $\Delta^{d-1}$ to another probability function $q$ in $\Delta^{d'-1}$ by the equation
\begin{equation}\label{eq:apply-channel}
    [\chn p](y) = \sum_x\chn(y|x)p(x).
\end{equation}
which simply formalizes the forward propagation from the input distribution $p$ to its output distribution $\chn p$. Thus, the stochastic process may be represented as a left-stochastic (column-stochastic) matrix $\chn$ in $\mathds{R}^{d'\times d}$ with $\chn_{y,x}:=\chn(y|x)$.
Finally, the joint probability function $P$ which lives in $\Omega' \times \Omega$ is given by $P(x,y):=\chn(y|x)p(x)$.

\subsubsection*{Bayesian Updating}

Central to this work is the process of updating one's description of an input given some output evidence (sometimes called retrodiction when $\chn$ is a forward process in time). That is, given an evidence distribution $q(y)$, which is some output sample of a characterized process $\chn(y|x)$, what can one infer (or update on earlier inferences) about the input distribution of that process? 

One might propose to simply mathematically invert the process and update one's description of the input to $\chn^{-1}q$. But for various reasons, this approach is neither mathematically general, physically sound nor practically robust \footnote{$\chn^{-1}$, in general, is not a physical stochastic process. If $q$ lies outside the image of $\chn$, $\chn^{-1} q$ could lie outside the probability simplex of $\Omega$. Such an approach is also not robust to making inferences when the process is singular, nor does it adhere to the strong practical principle of not neglecting prior information we had about the overall population for which the sample was obtained.}. Rather, the statistically normative approach centres around Bayes' rule:
\begin{equation}\label{eq:classical-retrodiction}
    \retd(x|y):=\frac{\chn(y|x)\rf(x)}{[\chn\rf](y)}
\end{equation}
where $\rf$ is referred to as the reference prior (or simply the prior), a probability distribution which captures our initial information or belief on the input, and $\chn\rf$ is its forward propagation. It is important to emphasize that, generally, different reference priors lead to different Bayes maps \cite{liu2025quantifying,liu2026-extendedpetz}, as this feature is pertinent to the results in this work. Unlike $\chn^{-1}$, $\retd$ always exists as a stochastic process as long as singularities are avoided in the denominator (which can always be done by adding small perturbations of the process \cite{aw2026thesis}). In matrix form, we can write Eq.~\eqref{eq:classical-retrodiction} as
\begin{equation}\label{eq:classical-retrodiction-matrix}
    \retd=D_{\rf}\chn^\tp D_{\chn\rf}^{-1},
\end{equation}
where $D_a$ is the diagonal matrix with entries $(D_a)_{i,i}=a_i$ for a given vector $a$. Now, with this Bayes map, one updates their prior $\rf$ to a new conclusion $p$:
\begin{equation}
    \rf \quad \mapsto \quad  p=\retd \,  q.
\end{equation}
This is how a genuine Bayesian or Jeffrey's update works \cite{jeffrey1990logic-BK,zhou2014belief-BK}. With that, we introduce a reverse-engineering of this: prior hacking.


\subsection{Classical Prior Hacking}\label{sec:classical-prior-hacking-introduce}
As mentioned, for proper Bayesian updates, a genuine reference prior $\rf_g$ induces a genuine conclusion $p_g = \ret{\chn}{\rf_g}\,q$ via Bayes' rule on $\chn$ and an evidence $q$. Now, one can invert this by starting with a \textit{fixed} conclusion $p_h$ which induces a frankly conceptually perverse ``prior'' $\rf_h$ instead. Consider Figure \ref{fig:prior-hacking-illus} for an illustration of this. 

Notice that, doxastically or sequentially speaking, $p_h$ and $\rf_g$ are the same thing: some initial belief. But while $\rf_g$ gets updated through Bayes' rule upon evidence $q$, $p_h$ is a fixed conclusion that we do not want to change our minds about, and yet also want to make it seem like a conclusion that emerges from rational Bayesian inference upon receiving some evidence. This inversion of Bayesian update prompts some $\rf_h$, which we call the hacking (or hacked) prior, for which $\ret{\chn}{\rf_h}q=p_h$. Prior hacking refers to the process of finding $\rf_h$ that solves this relation for some tuple $(\chn,p_h,q)$ of that process and some evidence and fixed conclusion. 

With that, we drop the subscripts to avoid encumbrance and state the problem of prior hacking concisely: \textit{given} a tuple $(\chn,p,q)$, under what conditions do we have $\rf$ such that
\begin{equation}\label{eq:classical-prior-hacking-problem}
    \retd \,q=p?
\end{equation}
From here, unless stated otherwise, $\rf$ is taken to be the hacking prior that solves Eq.~\eqref{eq:classical-prior-hacking-problem}.

First, note that we can rewrite Eq.~\eqref{eq:classical-prior-hacking-problem} to
\begin{equation}\label{eq:classical-prior-hacking-equation}
    p\oslash[\chn^\tp (q\oslash[\chn\rf])]=\rf.
\end{equation}
In this work, $\odot$ and $\oslash$ are the Hadamard or elementwise product and division, respectively. We also remark in passing that $\oslash$ might not be always be well-defined due to singularities. This becomes relevant when discussing when prior hacking is possible. The details of the Hadamard product and the derivation of Eq.~\eqref{eq:classical-prior-hacking-equation} can be also found in Appendix~\ref{app:hadamard}.

\subsubsection*{Key Theorems for Classical Prior Hacking}
We move to some key results that emerge from this problem. As it turns out, Eq.~\eqref{eq:classical-prior-hacking-equation} is equivalent to a known problem in Statistics known as the Sinkhorn problem, the Iterative Proportional Fitting (IPF), or matrix scaling, among other names \cite{sinkhorn1964relationship,sinkhorn1967diagonal,sinkhorn1967concerning,garg2018recent,idel2016review,menon1968matrix,brualdi1968convex,menon1969spectrum}. The problem is such: Given an initial matrix $X$, we wish to find $a$ and $b$, such that the matrix $M=D_aXD_b$,
has a row sum $u$ and column sum $v$. That is,
\begin{align}
    Me&=u \label{eq:ipf-pair-1} \\
    M^\tp e&=v, \label{eq:ipf-pair-2}
\end{align}
where $e=(1,1,\cdots,1)^\tp$. The details of this problem can be found in Appendix~\ref{app:ipf-sinkhorn}. Suffice to say, the problems are mathematically equivalent with the substitutions $(X,u,v,a,b)\longrightarrow(\chn,p,q,\rfd,\rf)$, where
\begin{equation}\label{eq:rfd-def}
    \rfd = q \oslash \chn\rf.
\end{equation}
One may see $\rfd$ as an expression of divergence between the actual result (captured by the evidence received as $q$) and the postulated result (that is, $\chn \rf$ the propagation of ones' reference prior). With this, we can take a known theorem and apply it to the question of classical prior hacking.

\begin{theorem}\label{thm:classical-prior-hacking}
    \textbf{(the Sinkhorn problem \cite{idel2016review,menon1968matrix,brualdi1968convex,menon1969spectrum})} Given a tuple $(\chn,p,q)$, where $\chn$ is a $d \times d'$ transition matrix and probability distributions $p \in \Delta^{d-1}$ and $q \in \Delta^{d'-1}
    $, the following are equivalent:
    \begin{enumerate}
        \item There exists $\rf$ such that
            \begin{equation}
                \retd \, q=p,
            \end{equation}
            where $\retd$ is the Bayes map of $\chn$, based on $\rf$.
        \item There exists positive vectors $\rfd$ and $\rf$ such that $D_\rfd\chn D_\rf$ has row sums $q$ and column sums $p$.
        \item There exists another matrix $B$ with row sums $q$ and column sums $p$, i.e. $Be=q$ and $B^\tp e=p$, with the condition
            \begin{equation}
                B_{y,x} > 0 \quad \forall(y,x) \text{ s.t } \chn_{y,x}>0.
            \end{equation}
        \item For every $Y \subset \{1,\dots,d'\}$ and $X \subset \{1,\dots,d\}$, such that $\chn_{Y^\bot X}=0$,
            \begin{equation}
                \sum_{y\in Y}q_y \geq \sum_{x\in X}p_x,
            \end{equation}
            with equality holding for $\chn_{YX^\bot}=0$ and $S^\bot$ denotes the complement of the set $S$.
    \end{enumerate}
\end{theorem}
\begin{proof}
    As mentioned, in Appendix~\ref{app:ipf-sinkhorn} we show that the first and second points are equivalent. A modern review of the proofs of the other points can be found in \cite{idel2016review}. The only difference here is that we require $\rf$ to be a normalised probability distribution. Since the two diagonal matrices $D_\rfd,D_\rf$ are related by a scalar constant, we can always normalise $\rf$ into a probability distribution.
\end{proof}

While the third statement is probably the most mathematically straightforward, the fourth statement may provide better physical intuitions. For every subset $Y$ and $X$ such that $\chn_{Y^\bot X}=0$ means that all transitions coming out from $x\in X$ must go to $Y$. Because every $p(x)$ must go to $X$, the total mass of all the $q(y)$ must be at least the same as the total mass of all the $p(x)$. Additionally, when $\chn_{YX^\bot}=0$, this means that $(Y,X)$ and $(Y^\bot,X^\bot)$ are decoupled from each other, and thus the total mass in each partition must be conserved. Finally, of note is that when all the entries of $\chn$ are positive, then the matrix $B_{i,j}:=\frac{q(i) \,p(j)}{S}$ where $S=\Sigma_i \, q(i)=\Sigma_j \, p(j)=1$ will always satisfies the second existence condition, meaning that it is prior hackable \footnote{Additionally, $\rf$ is unique if and only if $\chn$ is fully indecomposable. For more information on full indecomposability, we refer the readers to \cite{idel2016review,brualdi1966diagonal}.}.

Now, Theorem \ref{thm:classical-prior-hacking} deals with the prior-hacking with respect to a given tuple of $(\chn,p,q)$. A natural follow-up is to consider, for a given $\chn$, what conditions make it such that it is prior-hackable \textit{for all} $p,q$. This brings us to our next theorem, 

\begin{theorem}\label{thm:classical-prior-hacking-channel-forall} For transition matrices $\chn$, the following statements are equivalent:
\begin{enumerate}
\item $\retd$ is well defined for all $\rf$, 
 \item $\forall(y,x) :\chn_{y,x} >0$,  
\item $\chn$ is prior-hackable for all probability inputs $p$ and outputs $q$. 
\end{enumerate}
\end{theorem}
\begin{proof}
To avoid encumbering the main text, the proof can be found in Appendix \ref{app:thm-classical-hacking-forall}.
\end{proof}
With this, one might say that there is a vast embedded pluralism available to Bayesian reasoning. The set of transformations that omit prior hacking for some $p,q$ are measure-zero and are precisely those for which a Bayesian inversion is not always well defined. 

We note in passing that additional corollaries pertaining to concatenations of channels are included in Appendix \ref{app:concat-coros}. Corollary \ref{coro:concat-coro1} is a simple upshot that any concatenation of always prior-hackable channels is also always (i.e. for any $p,q$) prior-hackable (though the converse does not hold---prior-hackable concatenations may be decomposed into channels that are individually not always prior-hackable), while Corollary \ref{coro:concat-coro2} gives a definition of primitivity of Markov matrices \cite{herstein1954note-primitive-matrices} in explicitly Bayesian notions.

\subsubsection*{Obtaining the Hacked Prior}
While the existence of a hacking prior can be determined by Theorem~\ref{thm:classical-prior-hacking}, we have not yet described any method of finding a solution. To our knowledge, $\rf$ in Eq.~\eqref{eq:classical-prior-hacking-equation} cannot be written in a closed form solution. Fortunately, a corollary of Theorem~\ref{thm:classical-prior-hacking} is that an algorithm known as the RAS algorithm will converge to the solution of the problem given its existence \cite{sinkhorn1967concerning}. One expression of the algorithm can be obtained by simply taking Eq.~\eqref{eq:classical-prior-hacking-equation} and turning it into a fixed point iteration.
\begin{algorithm}\label{algo:classical}
    \textbf{(RAS algorithm)} The solution to Eq.~\eqref{eq:classical-prior-hacking-equation} can be obtained by randomly picking any initializing $\rf_0 \in \Delta^{d-1}$ with positive entries and iterating the following equation:
    \begin{equation}
    \rf_{i+1}=p\oslash[\chn^\emph{\tp}(q\oslash(\chn\rf_i))].
    \end{equation}
    As $i\rightarrow\infty$, $\rf_i$ will converge to the solution, if it exists.
\end{algorithm}
The way in which the algorithm works can be intuitively understood when we substitute back Eq.~\eqref{eq:rfd-def}. With this, an iteration of the algorithm is given by,
\begin{equation}
    \begin{tikzpicture}[
        node distance=1.0cm, 
        thick, 
        main node/.style={minimum size=0.8cm}
    ]
        \node[main node] (A) {$\rf_i$};
        \node[main node] (B) [right=of A] {$\chn\rf_i$};
        \node[main node] (C) [right=of B] {$\rfd_i$};
        \node[main node] (D) [right=of C] {$\chn^\tp\rfd_i$};
        \node[main node] (E) [right=of D] {$\rf_{i+1}$};

        \draw[->] (A) -- (B) node[midway, above] {$\chn$};
        \draw[->] (B) -- (C) node[midway, above] {$q\oslash\cdot$};
        \draw[->] (C) -- (D) node[midway, above] {$\chn^\tp$};
        \draw[->] (D) -- (E) node[midway, above] {$p\oslash\cdot$};

    \end{tikzpicture}.
\end{equation}
That is, it is an alternating series of propagation and counter propagation by the channel ($\chn$ and $\chn^\tp$) and element-wise scaling with the output and input states ($q$ and $p$).

\begin{figure}[t!]
    \centering
    \includegraphics[width=1.02\linewidth]{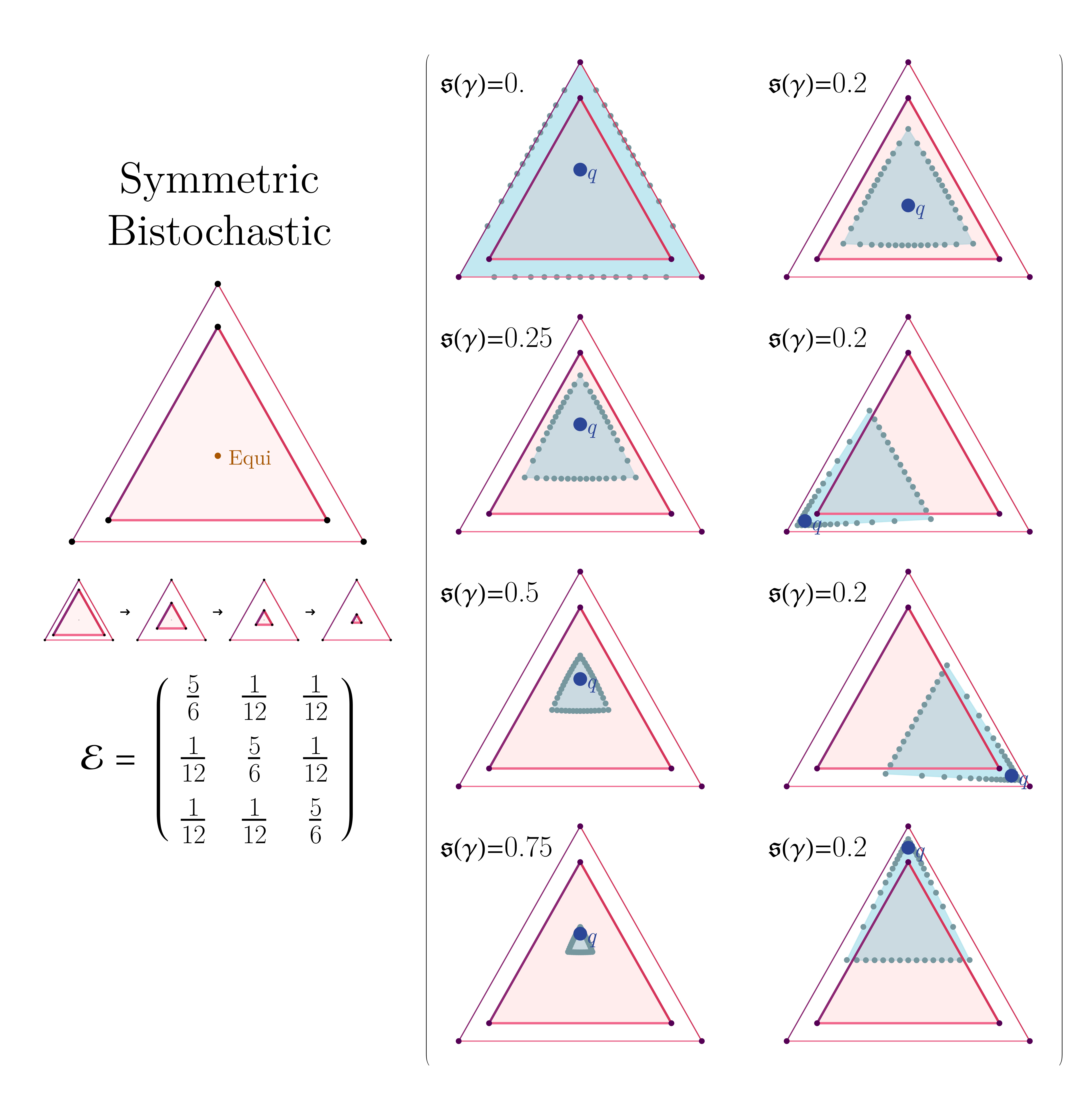}
    \caption{Image plots for a symmetric bistochastic channel acting on a trit space. For details on plot features, see Section \ref{sec:c-example-text}. Here, prior-hacking depends on $q$, in contrast to Figure \ref{fig:cerase}.}
    \label{fig:cdepolar}
\end{figure}

\begin{figure}[t!]
    \centering
    \includegraphics[width=1.02\linewidth]{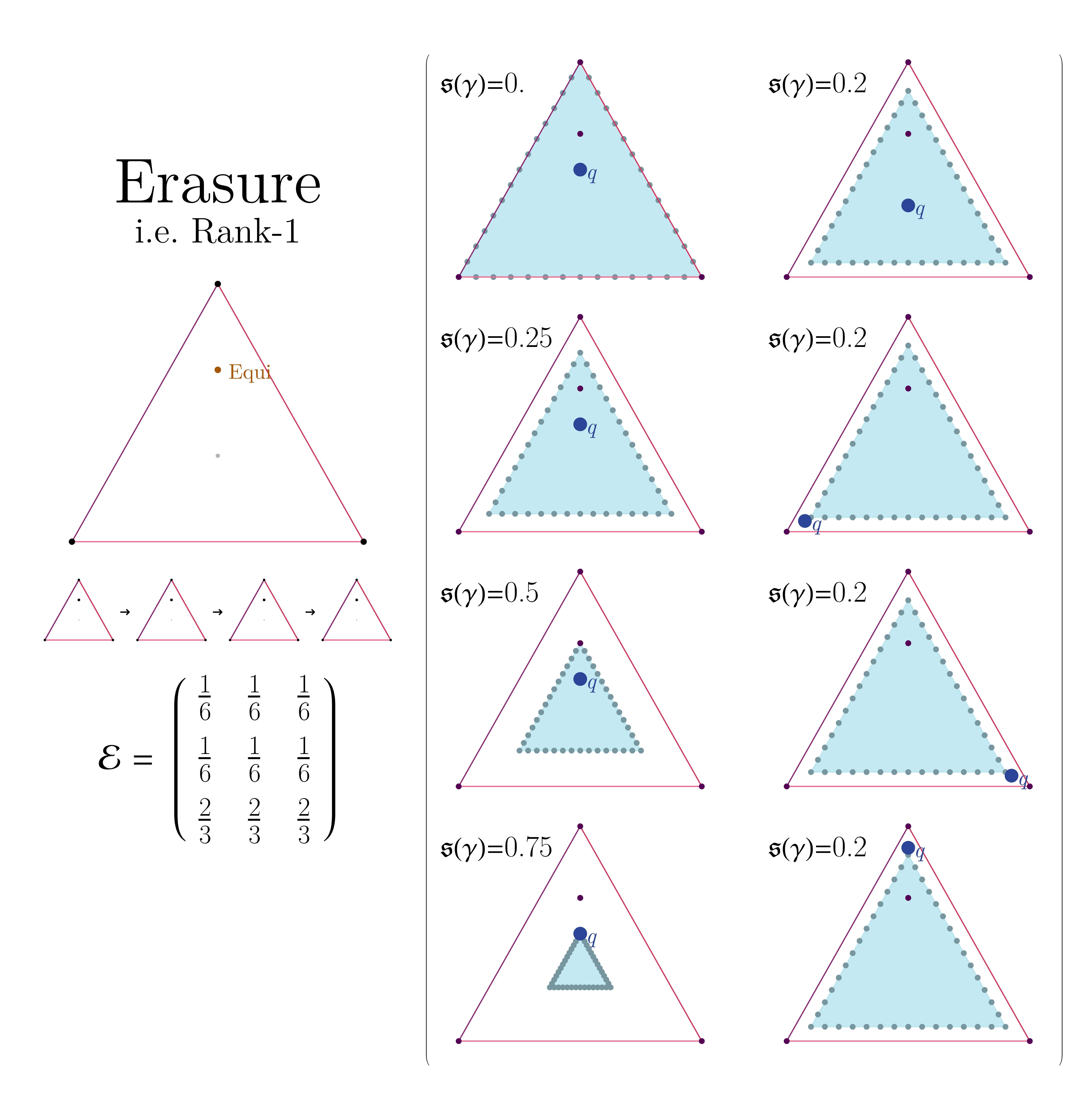}
    \caption{Image plots for a generic erasure channel acting on a trit space. For details, see Section \ref{sec:c-example-text}. Here, prior-hacking is shown to be independent of $q$, in contrast to Figure \ref{fig:cdepolar}.}
    \label{fig:cerase}
\end{figure}

\subsection{Prior Hacking for Notable Classes \\ of Classical Maps}\label{sec:c-example-text}
Before discussing our key results pertaining Schr\"odinger bridges, we briefly list key illustrative examples, with reference to figures, that showcase Theorem \ref{thm:classical-prior-hacking-channel-forall}'s geometric features. 
\begin{itemize}
    \item We first note the trivially reversible case for classical probability theory: \textit{permutation matrices}. These are matrices for which $\chn^{-1}=\chn^{\tp}$, being bistochastic matrices with only a single-$1$ entry for every row and column. They are \textit{never prior-hackable} except for the trivial case where $\chn p =q$. This is simply because $\forall\rf :\retd = \chn^{-1}$~\cite{AwBS, aw2024-tabletop}.
    
    \item On the extreme of irreversibility, we have \textit{erasure channels} (or discard-and-prepare channels), which immediately send every input to some stationary state $t$. That is, $\forall s \in \Delta^{d-1} : \chn s = t$. Aside from erasures to support-deficient outputs (for which Bayesian update is always undefined) \cite{aw2026thesis}, these channels are \textit{always} and trivially \textit{prior-hackable}. This is because Bayes' maps of erasure channels \textit{are} erasures to the prior $\forall q \in \Delta^{d-1} : \retd \, q = \rf$ \cite{AwBS, liu2025quantifying}. So setting $\rf =p$ automatically solves the prior hacking problem (see Figure \ref{fig:cerase}).
    
    \item \textit{Absorbing maps} are also worth mentioning as they are non-trivial exceptions to Theorem \ref{thm:classical-prior-hacking-channel-forall}'s conditions. Since details on these maps can be found in other resources \cite{liu2025quantifying,kemenysnell1969finite-absorbingmaps}, we simply note that these are maps that gradually send the probability weights of inputs inhabiting the ``transient space'' to a ``absorbing space'', where we have an identity channel $\chn(x|y)=\delta_{xy}$. For this reason, there are inevitably zero entries in the map, making some prior-hacking problems insoluble. This is illustrated in Figures \ref{fig:c01abs} and \ref{fig:c2abs}.
\end{itemize}

\subsubsection*{Details on Figures for Classical Examples}
In the figures mentioned above (as well as Figures \ref{fig:c-examples-in-app} included in the Appendix), the red filled triangles represent the image of the trit channels. The blue shaded regions are convex hulls of the image of $\retd \, q$ with $\rf$ as an input. In particular, we send in all the $\rf$ of informational entropy $\mathfrak s(\rf)$ and plot the outgoing $\retd \, q$, represented by the perimeter of the region. This may be referred to as the image of the hacking channel given $\mathfrak s(\rf)$. Meanwhile, the larger dark blue dot designates the evidence received $q$ and the dark orange dot in the leftmost plots designates the ``equilibrium'' state of $\chn$ (i.e. the channel's fixed point).

Note how Figures \ref{fig:cdepolar}, \ref{fig:cerase} and \ref{fig:cgen} have a full images for the hacking channel for pure priors (i.e. the blue shaded regions for $\mathfrak s(\rf)=0$ capture the whole simplex). This captures their total prior-hackability as per Theorem \ref{thm:classical-prior-hacking-channel-forall}. Meanwhile, Figures \ref{fig:c01abs}, \ref{fig:c2abs} and \ref{fig:cblocksym} always have incomplete images for their respective hacking channels, as they do not fulfil the conditions of Theorem \ref{thm:classical-prior-hacking-channel-forall}.  

\begin{figure}[t!]
    \centering
    \includegraphics[width=1.02\linewidth]{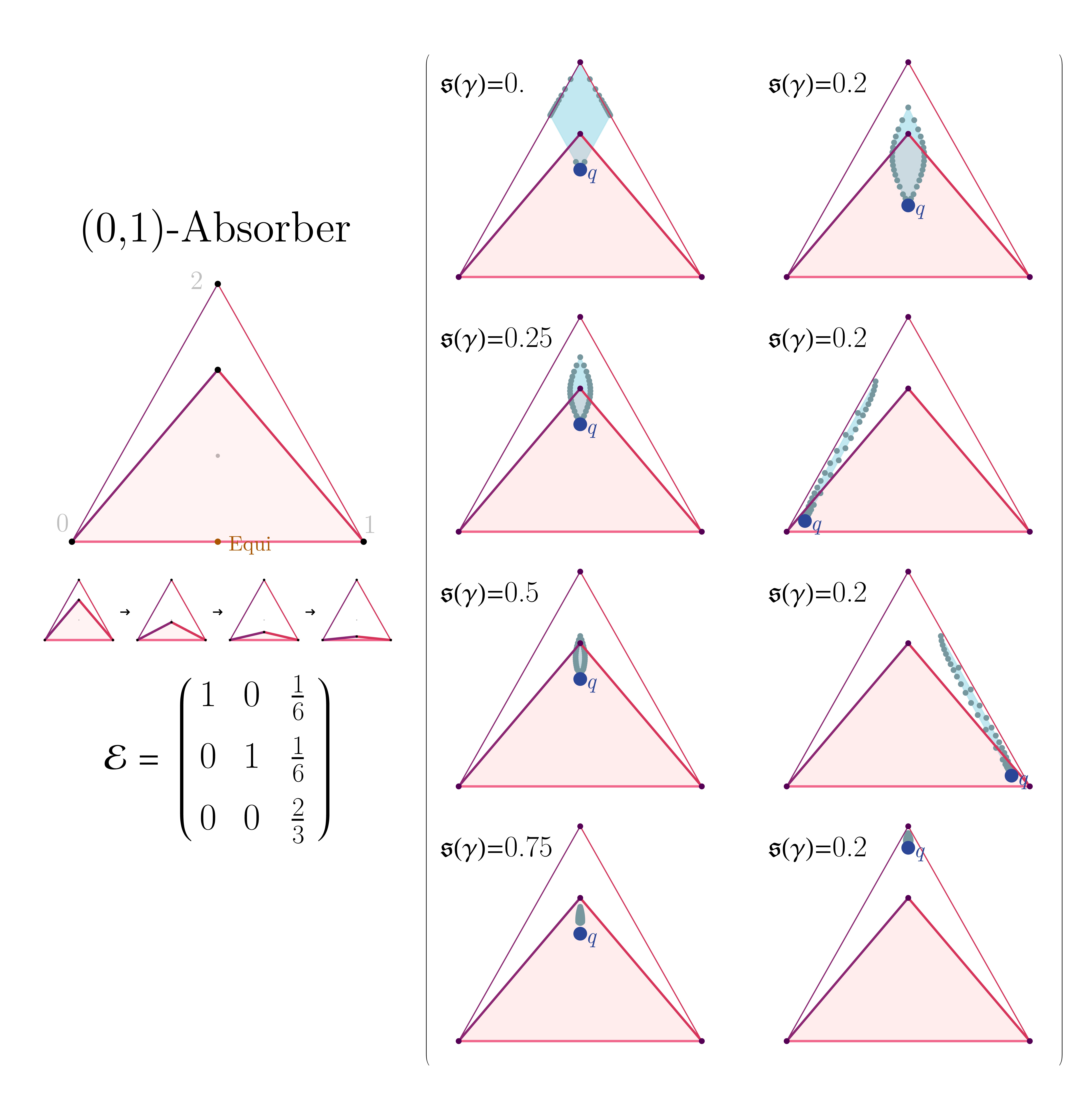}
    \caption{Image plots for a $(0,1)$-absorber channel acting on a trit space. For details, see Section \ref{sec:c-example-text}. Notice how one can only perform a Bayesian update toward states with weights on the $2$-state that are higher than that of $q$. Thus, prior-hacking is not always possible for any given $p,q$, as per Theorem \ref{thm:classical-prior-hacking-channel-forall}.}
    \label{fig:c01abs}
\end{figure}

\subsection{Classical Schrödinger Bridges \& Prior Hacking} \label{ssec:classical-SB}

The Sinkhorn problem has been rediscovered or found to be equivalent to other problems in various fields, such as transportation planning \cite{evans1970some,schneider1990comparative,de2024modelling}, contingency tables \cite{fienberg1970iterative,plane1982information}, social accounting matrices \cite{pyatt1985social,thorbecke2000use}, elections \cite{balinski1989algorithms}, preconditioning matrices in linear algebra \cite{osborne1960pre,bradley2010algorithms}, machine learning \cite{cuturi2013sinkhorn}, and even local hidden variable models \cite{aaronson2005quantum}. However, what is of particular interest to our quantum information setting is referred to as the Schrödinger bridge (SB) problem \cite{fortet1940resolution-SB-Org,leonard2013survey-SB-survey,chen2016relation-SB-Intro,chetrite2021schrodinger-SB-Org,orland2025schrodinger-SB-gen}, especially its discretised single step formulation \cite{georgiou2015positive}. This brings us to the primary connection of this work, between prior hacking and these objects in statistical physics. We begin first by reviewing key features of the Schrödinger bridge problem.

\subsubsection*{The Schrödinger Bridge Problem}

SB problems consist of an observation of initial and final distributions of particles in a system, $p(x)$ and $q(y)$, as well as a prior joint distribution $A(x,y)$. However, at least for non-trivial cases, $\sum_{y} A(x,y) \neq p(x)$ and $\sum_{x} A(x,y) \neq q(y)$. The task is to find a distribution $B$, the Schrödinger bridge, that is as close as $P$ whose marginals still satisfy both distributions $p,q$. The closeness here is defined by the KL-divergence.
\begin{equation}\label{eq:optimisation-problem-sb}
\begin{array}{r c l}
    B = & \displaystyle \argmin_{B'} & \textsf{D}(B'\| A) \\
      = & \displaystyle \argmin_{B'} & \displaystyle \sum_{x,y} B'(x,y) \log\left(\frac{B'(x,y)}{A(x,y)}\right). \\ 
        & \text{subject to}                & \displaystyle \sum_{x} B'(x,y)=q(y) \\
        &                            & \displaystyle \sum_{y} B'(x,y)=p(x)
\end{array}
\end{equation}
The solution to this problem, it turns out, can be found by a method equivalent to the Sinkhorn problem. First, we have to find $\rfd$ and $\rf$ such that the following equations are satisfied:
\begin{align}
    &[\chn^\tp\rfd](x) \, \rf(x) = p(x) \label{eq:sb-pair-1-function}\\
    &[\chn\rf](y) \, \rfd(y) = q(y), \label{eq:sb-pair-2-function}
\end{align}
where
\begin{align}
    &[\chn^\tp\rfd](x) = \sum_{y} \chn(y|x)\rfd(y) \\
    &[\chn\rf](y) = \sum_{x} \chn(y|x)\rf(x).
\end{align}
Here $\chn(y | x)=A(x,y)/\sum_{y}A(x,y)$ is the transition probability induced by the prior distribution. Note that, firstly, in this setting neither $\rfd$ nor $\rf$ has to be a probability distribution. Instead, $\rf$ is called the forward scaling potential and $\rfd$ the backward scaling potential. However, they are still related by Eq.~\eqref{eq:rfd-def} that was defined for the prior hacking problem, and are thus not independent of each other. Secondly, the marginal $\sum_{y}A(x,y)$ does not matter in this problem (see Appendix \ref{app:optimisation-sb-quantum}). 

The solution (i.e. the SB) would then be given by
\begin{equation}
    B (x,y) = \rf(x) \chn(y|x) \rfd(y),
\end{equation}
and the optimum transition probability $\chnd(y|x)=B(x,y)/p(x)$ is given by
\begin{equation}\label{eq:SB-transition-entries}
    \chnd(y|x) = \chn(y|x) \frac{\rfd(y)}{\chn^\tp\rfd(x)}.
\end{equation}

\subsubsection*{Connection to Prior Hacking}
One immediately notices Eq.~\eqref{eq:SB-transition-entries}'s resemblance to Bayes' rule as per Eq.~\eqref{eq:classical-retrodiction}. Indeed, as alluded to already, the SB problem has a surprisingly close relation to the prior hacking problem. To see this, we first write both Eqs.~\eqref{eq:sb-pair-1-function} and \eqref{eq:sb-pair-2-function} in matrix notation as
\begin{align}
    &(\chn^\tp\rfd) \odot \rf = p, \label{eq:sb-pair-1}\\
    &\rfd \odot (\chn\rf) = q. \label{eq:sb-pair-2}
\end{align}
If the objects are nonzero, these give nothing but Eq.~\eqref{eq:classical-prior-hacking-equation}, which is simply the prior hacking problem.

With this we note some key connections between what has been discussed. Firstly, solving the classical, single-step discrete SB problem is mathematically equivalent to solving the Sinkhorn problem and the prior hacking problem. In fact, the Eqs.~\eqref{eq:sb-pair-1} and \eqref{eq:sb-pair-2} are equivalent to the Eqs.~\eqref{eq:ipf-pair-1} and \eqref{eq:ipf-pair-2} that form the constraints that define the Sinkhorn problem. Secondly, while $\rf$ in the SB is not necessarily a probability distribution, it can always be normalised into a vector due to the freedom in the scalar factor of $\rf$ and $\rfd$. 

This naturally leads to our third point. While the Sinkhorn and the SB problem invokes $\rfd$ and $\rf$ under a rather abstract interpretation, one may argue that the prior-hacking scenario endows additional conceptual clarity for these objects. $\rf$ under scalar normalization is nothing but the reference prior, whereas $\rfd$ of Eq.~\eqref{eq:rfd-def} is the aforementioned divergence between the evidence $q$ and the postulated result $\chn\rf$.


Finally, we note that via Eq.~\eqref{eq:SB-transition-entries} the SB transition matrix may be expressed as
\begin{equation}\label{eq:classical-SB-transit}
    \chnd = D_\rfd\chn D_{\chn^\tp\rfd}^{-1}.
\end{equation}
As alluded, there is much resemblance to Bayes maps as written in Eq.~\eqref{eq:classical-retrodiction-matrix}. However, a key difference is that a transpose is not taken on $\chn$ and so $\chnd$ still maps $\Omega \longrightarrow \Omega'$, as opposed to $\Omega' \longrightarrow \Omega$ for $\retd$. Rather, $\chnd$ is an adjustment of the forward map $\chn$ that gives the correct output $q$ given an initial input $p$. This brings us to a key observation of this work:
\begin{theorem}\label{thm:CSB-is-CPH}
    For a Schr\"odinger bridge or prior hacking problem characterized by $(\chn, p,q)$, the Bayesian inversion of $\chn$ based on the hacked prior $\rf$ and the Bayes inverse of the Schr\"odinger bridge based on $p$ are always the same object:
    \begin{equation}\label{eq:CSB-is-CPH}
    \ret{\chnd}{p} = \retd
    \end{equation}
\end{theorem}
\begin{proof}
    Applying Bayesian inversion to the adjusted forward channel $\chnd$ with the prior $p$, we get:
\begin{align}
    \ret{\chnd}{p} &= D_{p} \chnd^\tp D^{-1}_{\chnd p} \nonumber \\
    &= D_{p} D_{\chn^\tp\rfd}^{-1}\chn^\tp D_\rfd D^{-1}_{\chnd p} \nonumber \\
    &= D_{p} D_{\chn^\tp\rfd}^{-1}\chn^\tp D_\rfd D^{-1}_{q} \nonumber \\
    &= D_{\rf} \chn^\tp D_{\chn \rf}^{-1},
\end{align}
where we have applied Eqs.~\eqref{eq:sb-pair-1} and \eqref{eq:sb-pair-2} and $\chnd p=q$ is given by the definition of the SB. This gives nothing but Eq.~\eqref{eq:CSB-is-CPH}.
\end{proof}
This equivalence is noteworthy. It may be understood as in terms of two routes seeking one same correspondence between $p$ and $q$. $\ret{\chnd}{p} \, q = p,$ is guaranteed because Bayesian inference always recovers the initial guess if it was accurate (i.e. $\forall (\chn,\rf) : \retd \chn\gamma = \rf$) \cite{petzisking2022axioms}. Meanwhile, the goal of prior hacking is also this match via $ \hat{\chn}_{\rf}\, q = p$, though for very different reasons, as $q\neq \chn\rf$ and $\rf \neq p$. In the Schrödinger bridge problem, instead of hacking or modulating the \textit{prior} $\rf$ to match $p$ to $q$, one hacks or modulates the \textit{process} $\chn$ to do the same.

Finally, we note that once a prior hacking problem $(\chn,p,q)$ is solved, one would also have trivially solved for $\chnd$ via Eq.~\eqref{eq:CSB-is-CPH} and \eqref{eq:classical-retrodiction}:
\begin{equation}\label{eq:CSB-from-CPH}
    \forall(x,y): \; \chnd(y|x)=\retd(x|y)\frac{q(y)}{p(x)}.
\end{equation}

\begin{figure}[h]
    \centering
    \begin{tikzpicture}[
        node distance=2cm, 
        thick, 
        main node/.style={minimum size=1cm}
    ]
        \node[main node] (A) {$\chn$};
        \node[main node] (B) [right=of A] {$\chnd$};
        \node[main node] (C) [below=of A] {$\ret{\chn}{\rf'}$};
        \node[main node] (D) [below=of B] {$\retd=\ret{\chnd}{p}$};

        \draw[->] (A) -- (B) node[midway, above] {Schrödinger Bridge};
        \draw[<->] (D) -- (B) node[midway, right, align=center] {Bayes Inference\\with prior $p$};
        \draw[->] (C) -- (D) node[midway, below, align=center] {Prior Hacking\\$\rf'\mapsto\rf$};
        \draw[<->] (A) -- (C) node[midway, left, align=center] {Bayes Inference\\with arbitrary \\ prior $\rf'$};

    \end{tikzpicture}
    \caption{The relationship between prior hacking and Schrödinger bridge (Theorem \ref{thm:CSB-is-CPH}). The map obtained through prior hacking, $\retd$, is the same map as the reverse map of the Schrödinger bridge with reference prior $p$, $\ret{\chnd}{p}$. With $p \to \rho$ and Bayes' rule upgraded to the Petz recovery, this commutativity diagram applies also to the quantum regime (Theorem \ref{thm:QSB-is-QPH}), but for the inference-consistent QSB only.}
\end{figure}
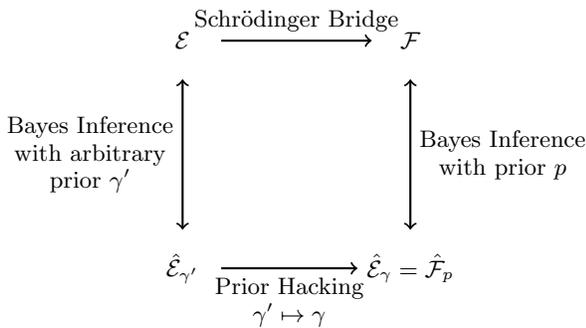

\section{Quantum Setting}\label{sec:quantum}
\subsection{Quantum Prior Hacking}\label{ssec:quantum-setting}
With the classical setting discussed, we move to the prior hacking in the quantum regime. Here, probability distributions $p,q$ are upgraded to quantum states $\rho,\omega$. These are operators in the Hilbert space $\mathcal{H}^d$ represented by density matrices in $\mathds{C}^{d\times d}$, with the constraints that $\Tr\rho=1$ and $\rho \succeq 0$. Moreover, the forward channel $\chn$ is upgraded to a quantum channel, which is a map $\chn[\cdot]:\mathcal{H}^d \longrightarrow \mathcal{H}^{d'}$ that is completely positive and trace preserving (CPTP). Our choice of the quantum analogue to Bayesian inversion Eq.~\eqref{eq:classical-retrodiction} is the Petz recovery map \cite{petz,petz1,barnum-knill,petzisking2022axioms, liu2026-extendedpetz,aw2026thesis}, given by:
\begin{equation}\label{eq:petz-map}
    \retd[\omega]=\sqrt{\rf}\chn^\dag\left[\frac{1}{\sqrt{\chn[\rf]}}\omega\frac{1}{\sqrt{\chn[\rf]}}\right]\sqrt{\rf}.
\end{equation}
The task of prior hacking is still operationally the same: given an evidence state $\omega$ and a conclusion $\rho$, is it possible to find a reference prior state $\rf$ (now a density operator) such that $\retd[\omega]=\rho$? To this, we state the quantum analogue of Theorem \ref{thm:classical-prior-hacking-channel-forall}:

\begin{theorem}\label{thm:quantum-surjective}
    For any pair of states $\rho,\omega \in \mathcal{H}^d$ and any channel $\chn: \mathcal{H}^d \longrightarrow \mathcal{H}^{d}$ that always maps to full rank states, i.e., $\forall\rho' \; \text{rank}(\chn[\rho'])=d$, prior hacking is always possible, i.e., $\exists \rf \; \text{s.t.} \; \retd[\omega]=\rho$.
\end{theorem}
\begin{proof}
    To avoid encumbering the main text, the proof can be found in Appendix~\ref{app:thm-quantum-surjective}.
\end{proof}

The condition that $\chn[\cdot]$ is always full rank is akin to $\chn$ having strictly positive entries in the classical setting. In both cases, these conditions are needed to make sure that the reverse map $\retd$ is always well defined regardless of the choice of the reference prior $\rf$. We will discuss some example of quantum channels that do not satisfy this in Sec.~\ref{ssec:quantum-examples}.

\subsubsection*{Obtaining the Hacked Prior}
Though the quantum solution to prior hacking is guaranteed to exist due to Theorem~\ref{thm:quantum-surjective}, finding the solution is not straightforward. One possible way to do it is to derive a fixed point iteration similar to Algorithm~\ref{algo:classical}. First, however, we have to define $\rfd$ for the quantum scenario:
\begin{equation} 
\label{eq:rfd-def-qm}
    \rfd=\frac{1}{\sqrt{\chn[\rf]}}\omega\frac{1}{\sqrt{\chn[\rf]}}.
\end{equation}

\begin{algorithm}\label{algo:quantum}
    The solution to quantum prior hacking for a channel $\chn$, evidence $\omega$, and target $\rho$ can be obtained by randomly picking a full-rank $\rf_0$ from the space of states and iterating the following equation:
\begin{equation}
    \rf_{i+1}=\left[\sqrt{\rho}\left(\frac{1}{\sqrt{\rho}}(\chn^\dag\left[\rfd_i\right])^{-1}\frac{1}{\sqrt{\rho}}\right)^\frac{1}{2}\sqrt{\rho}\right]^2,
\end{equation}
where 
\begin{equation}\label{eq:QPH-algo-eq}
    \rfd_i=\frac{1}{\sqrt{\chn[\rf_i]}}\omega\frac{1}{\sqrt{\chn[\rf_i]}}.
\end{equation}
As $i\rightarrow\infty$, $\rf_i$ will converge to the solution.
\end{algorithm}


The derivation of the iteration can be found in Appendix~\ref{app:quantum-ipf-sinkhorn}. As with the classical version (Algorithm~\ref{algo:classical}
), this algorithm still alternates between propagation and counter propagation steps by the channel ($\chn[\cdot]$ and $\chn^\dag[\cdot]$) and scaling steps with the output and input states (represented by the operations $\mathcal{D}_\omega(\cdot)$ and $\mathcal{D}'_\rho(\cdot)$).

Step-by-step, it can be written as:
\begin{equation*}
    \begin{tikzpicture}[
        node distance=0.8cm, 
        thick, 
        main node/.style={minimum size=1.0cm}
    ]
        \node[main node] (A) {$\rf_i$};
        \node[main node] (B) [right=of A] {$\chn[\rf_i]$};
        \node[main node] (C) [right=of B] {$\rfd_i$};
        \node[main node] (D) [right=of C] {$\chn^\tp\rfd_i$};
        \node[main node] (E) [right=of D] {$\rf_{i+1}$};

        \draw[->] (A) -- (B) node[midway, above] {$\chn[\cdot]$};
        \draw[->] (B) -- (C) node[midway, above] {$\mathcal{D}_\omega(\cdot)$};
        \draw[->] (C) -- (D) node[midway, above] {$\chn^\dag[\cdot]$};
        \draw[->] (D) -- (E) node[midway, above] {$\mathcal{D}'_\rho(\cdot)$};

    \end{tikzpicture}
\end{equation*}
where
\begin{align}
\mathcal{D}_\omega(\cdot)
&= \frac{1}{\sqrt{\cdot}} \, \omega \, \frac{1}{\sqrt{\cdot}} \label{eq:quantum-algo-step-1}\\
\mathcal{D}'_\rho(\cdot)
&= \left[
\sqrt{\rho}
\left(
\frac{1}{\sqrt{\rho}\,\cdot\,\sqrt{\rho}}
\right)^{1/2}
\sqrt{\rho}
\right]^2 . \label{eq:quantum-algo-step-2}
\end{align}
Notably, unlike in the classical case, $\mathcal{D}_\omega$ and $\mathcal{D}'_\rho$ are not symmetric with each other.
\begin{center}
\begin{figure}[ht!]
    \centering
    \includegraphics[width=0.79\linewidth]{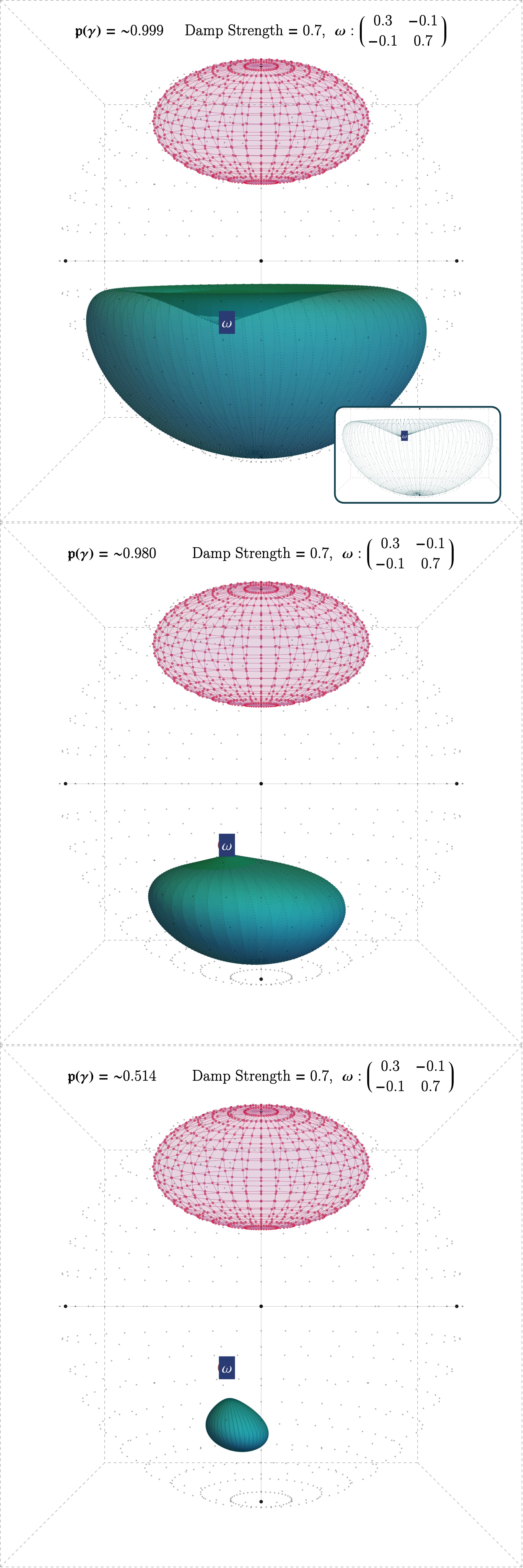}
    \caption{Blue plots correspond to (generally non-convex) images of prior hacking maps $\mathcal{X}_{\chn,\omega}(\rf)$ as per Eq.~\eqref{eq:q-prior-hack-channel} for an amplitude damping channel, of varying purity $\mathfrak{p}(\rf)$ for input $\rf$,  keeping the same damping strength and arbitrary evidence $\omega$. Red plots are the image of the amplitude damping channel itself on the Bloch sphere.}
    \label{fig:amp-damp}
\end{figure}
\end{center}
\begin{center}
\begin{figure*}[ht]
    \centering
    \includegraphics[width=0.98\textwidth]{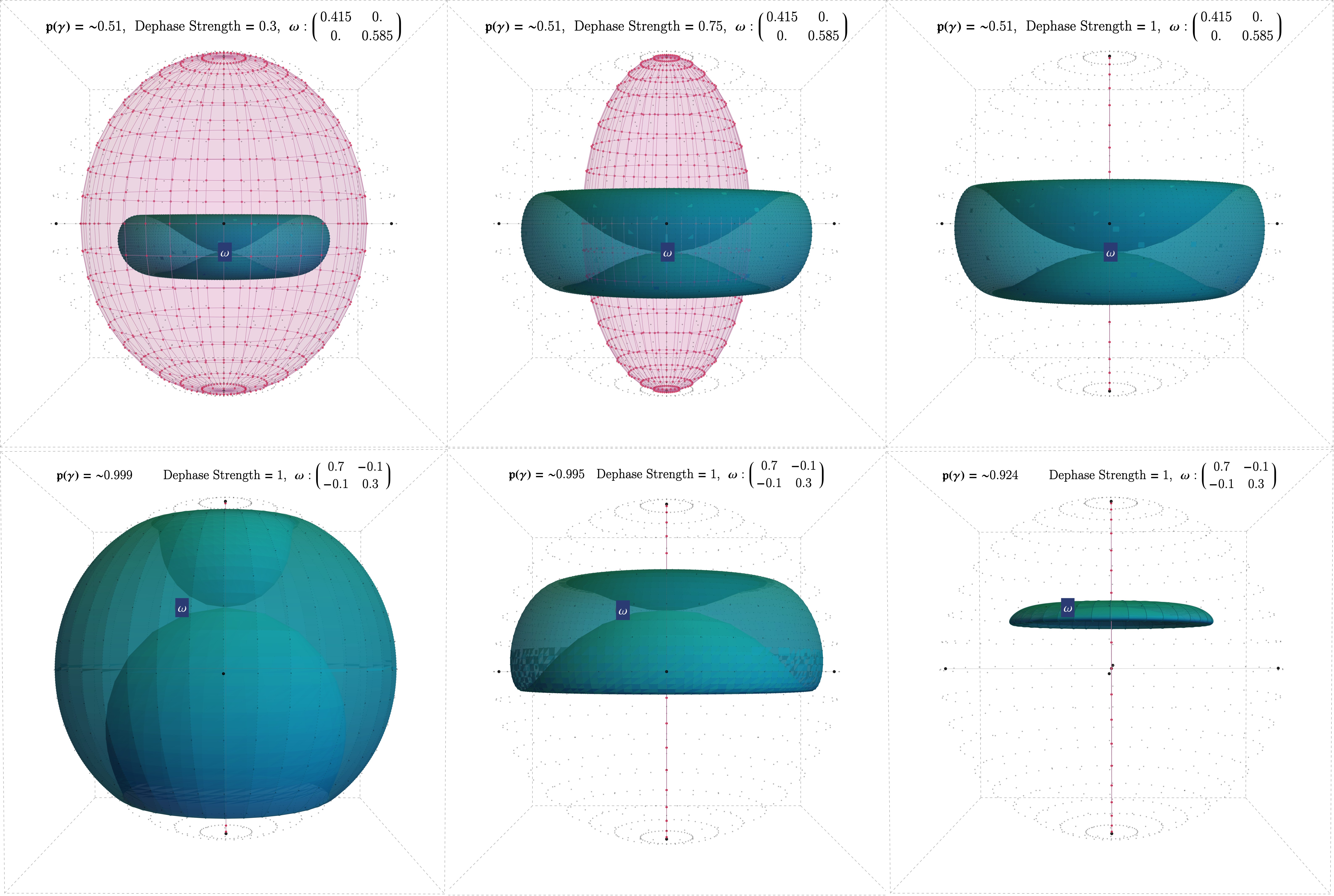}
    \caption{Blue plots correspond to (generally non-convex) images of the prior hacking maps $\mathcal{X}_{\chn,\omega}(\rf)$ as per Eq.~\eqref{eq:q-prior-hack-channel} for dephasing channel. The red plots correspond to the image of the dephasing channels themselves. The top row features varying dephasing strengths. The bottom row varies the purity $\mathfrak{p}(\rf)$ of input $\rf$, for a totally dephasing channel and a fixed generic evidence $\omega$. Notice how the hacking image bends in on itself toward the plane of dephasing containing the evidence, in adherence to Theorem \ref{thm:dephasing-channel}}
    \label{fig:dphs-damp}
\end{figure*}
\end{center}
\subsection{Notable Qubit Channel Cases}\label{ssec:quantum-examples}
As with the classical setting, we briefly run through some illustrative examples of quantum prior hacking. Since this work's focus is not on specific channels, we will only briefly state some notable observations, leaving more specific theorems for future exploration.
\begin{itemize}
    \item \textit{Unitaries} are analogous to permutation matrices in the classical regime, being reversible. Expectedly, non-trivial prior hacking is always impossible in this case since $\forall\rf: \ret{\mathcal{U}}{\rf} = \mathcal{U}^\dagger$ \cite{AwBS, liu2025quantifying}.
    \item As with classical erasures, \textit{quantum erasures} send all input density operators to the same density operator. Applying the Petz recovery map on such discard-and-prepare channels also yield quantum erasures, but to the reference state \cite{AwBS, liu2025quantifying}. Thus, quantum prior hacking here is the same as classical prior hacking---since all information is erased, hacking is always trivial with $\rf = \rho$. 

\item \textit{Amplitude damping} channels for qubit systems are akin to absorbing maps for classical stochastic matrices. Since some choice of vacuum state $\ket{0}\!\bra{0}$ is sent to itself, these channels do not satisfy the conditions of Theorem \ref{thm:quantum-surjective}. This excludes prior hacking for classical evidence toward conclusions of greater weights in $\ket{0}\!\bra{0}$. Consider Figure \ref{fig:amp-damp} for additional geometric intuition. 
    
\item The depolarising channel,
\begin{equation*} 
    \depolarise[\rho]=\lambda\frac{\mathds{I}}{2}+(1-\lambda)\rho,
\end{equation*}
interpolates between an unitary (identity) channel ($\lambda=0$), where non-trivial prior hacking is impossible, and an erasure channel ($\lambda=1$), where it is trivial. For intermediate $\lambda$, the Petz map depends continuously on both the channel and the reference state. As the reference state varies, its image deforms continuously from a single point to the full Bloch ball. By continuity, no region can remain inaccessible, implying surjectivity and hence the possibility of prior hacking (Theorem~\ref{thm:quantum-surjective}). For details, see Appendix \ref{app:qm-chan}.

\item The dephasing channel $\dephase$ is defined by the Kraus operators
\begin{equation}
    Z_0 = \begin{pmatrix} 1 & 0 \\ 0 & \sqrt{1-\lambda} \end{pmatrix}, \quad
    Z_1 = \begin{pmatrix} 0 & 0 \\ 0 & \sqrt{\lambda} \end{pmatrix}.
\end{equation}
It suppresses off-diagonal terms and, unlike the depolarising channel, fails the full-rank condition of Theorem~\ref{thm:quantum-surjective}. Specifically, computational basis states are fixed points and remain rank-deficient, rendering the Petz map ill-defined for such references and preventing prior hacking to pure basis states. Nevertheless, prior hacking can still occur in restricted cases. This is most apparent in the completely dephasing limit ($\lambda=1$), where the channel projects onto diagonal states. Notably, for any evidence $\omega$ and toward a decoherent conclusion $\rho$, prior hacking is possible if and only if $\rho,\omega$ shares the same weights in the computational basis. For details, consider Figure \ref{fig:dphs-damp} and Theorem \ref{thm:dephasing-channel} in Appendix \ref{app:qm-chan}.
\end{itemize}

\subsection{Quantum Schrödinger Bridges \& Prior Hacking}\label{ssec:quantum-SB}

\subsubsection*{Generic Candidates for Quantum Schrödinger Bridges}

Given Theorem \ref{thm:CSB-is-CPH}, one may be curious if something similar obtains for quantum prior hacking and some quantum edition of Schrödinger bridges. It turns out this is the case, but with important caveats. We first note that there are a number of works on quantum Schrödinger bridges (QSBs) \cite{pavon2002quantum-SB-QM,pavon2010discrete,ticozzi2010time,movilla2025quantum-SB-QM,friedland2017schrodinger}, but here we focus on the discrete single-step formulation of \cite{georgiou2015positive}. In analogue to the classical situation delineated in Section \ref{ssec:classical-SB}, we are given a CPTP map $\chn$ as an initializing reference quantum dynamics to be updated based on observations $\rho$ and $\omega$ corresponding to initial and terminal states respectively. The goal is to find a new evolution $\chnds$ that is close to the initial one but with a key feature that
\begin{equation}
    \chnds[\rho]=\omega.
\end{equation}
Moreover, since it is trace-preserving it must also obey $\chnds^\dag[\mathds{1}]=\mathds{1}$. Now, QSBs as is defined in \cite{georgiou2015positive}, is \textit{not} derived from an optimisation of some divergence such as in the classical case (see Eq.~\eqref{eq:optimisation-problem-sb}). This is due to the ambiguity on what entropic quantities should be optimised (consider Appendix \ref{app:optimisation-sb-quantum}). 

Instead, the QSB is constructed (with reference to the classical form Eq.~\eqref{eq:classical-SB-transit}) as the quantum analogue of the multiplicative functional transformation of the initial quantum channel. In other words, invertible matrices $\poi$ and $\pof$ \footnote{In the more general QSB problem, where $\poi$ and $\pof$ do not need to be Hermitian, the problem is equivalent to an operator scaling problem with specified marginals \cite{gurvits2004classical,garg2020operator,franks2018operator}. As a result, $\chnds$ could be solved by a different algorithm with proven convergence (see \cite{franks2018operator} for more information).}, along with Hermitian matrices $\rfsb$ and $\rfdsb$ must be obtained such that we obtain the bridge
\begin{equation}\label{eq:qm-SB-channel-standard}
    \chnds [\cdot] = \pof\chn[\poi^{-1}\cdot\poi^{-\dag}]\pof^\dag.
\end{equation}
Moreover, the constraints that must be satisfied are,
\begin{align}
    \rho &= \poi \rfsb \poi^\dag \label{eq:qsb-constraint-1} \\
    \omega &= \pof \chn[\rfsb] \pof^\dag \label{eq:qsb-constraint-2}\\
    \chn^\dag[\rfdsb] &= \poi^\dag \poi \label{eq:qsb-constraint-3}\\
    \rfdsb &= \pof^\dag \pof \label{eq:qsb-constraint-4}.
\end{align}
Whenever these conditions obtain, $\chnds$ gives a QSB. 

\subsubsection*{An Aside on Georgiou--Pavon Conjecture}\label{ssec:qsb-conjecture}

Before establishing the connection between Eq.~\eqref{eq:qm-SB-channel-standard} and quantum prior hacking, we take a slight detour pertaining to a particular choice of $\poi$ and $\pof$ potentials in the literature. In \cite{georgiou2015positive}, the authors conjectured that the bridge can always be solved for a channel that is always \textit{positivity improving}, that is $\chn(\rho)>0$ for all possible $\rho$. This is nothing but our condition in Theorem~\ref{thm:quantum-surjective}. With this, we state their conjecture
\begin{conjecture}\label{conj:gp15}
    \textbf{(Georgiou and Pavon \cite{georgiou2015positive})} Given a positivity improving CPTP map $\chn$ and two density matrices $\rho$ and $\omega$, there exists invertible matrices $\poi$ and $\pof$ and Hermitian matrices $\rfsb$ and $\rfdsb$ that satisfy Eq.~\eqref{eq:qsb-constraint-1} to \eqref{eq:qsb-constraint-4}. In particular, we can also set $\poi = \sqrt{\chn^\dag[\rfdsb]}$ and $\pof = \sqrt{\rfdsb}$.
\end{conjecture}
Note that after taking both $\poi$ and $\pof$ to be Hermitian, the QSB will be in the form of
the QSB of Eq.~\eqref{eq:qm-SB-channel-standard} will be in the form of
\begin{equation}\label{eq:qm-SB-channel-standard-hermitian}
    \chndh[\cdot] = \sqrt{\rfdsb} \left(\chn\left[\frac{1}{\sqrt{\chn^\dag[\rfdsb]}}\cdot \frac{1}{\sqrt{\chn^\dag[\rfdsb]}}\right] \right)\sqrt{\rfdsb},
\end{equation}
which mirrors the Petz map closely. Meanwhile, the relation between $\rfsb$ and $\rfdsb$ can now be written as \footnote{Note that in the original paper, this equation was written as $\rfsb=\sqrt{\chn^\dag[\rfdsb]} \omega \sqrt{\chn^\dag[\rfdsb]}$. We take this as a misprint.}
\begin{equation}\label{eq:rfd-def-qsb}
    \rfsb=\frac{1}{\sqrt{\chn^\dag[\rfdsb]}}\rho\frac{1}{\sqrt{\chn^\dag[\rfdsb]}}.
\end{equation}
This equation can be seen as the QSB analogue to Eq.~\eqref{eq:rfd-def-qm} for the Petz map. Moreover, the authors have also suggested an algorithm to solve the QSB when taking $\poi$ and $\pof$ to be Hermitian. Expectedly, this bears a striking resemblance to Algorithm~\ref{algo:quantum}.
\begin{algorithm}\label{algo:qsb}
    The solution to quantum Schrödinger bridge for a channel $\chn$, evidence $\omega$, and target $\rho$ can be obtained by randomly picking a full-rank $\rf_0$ from the space of states and iterating the following equation:
\begin{equation}
    \rfdsb_{i+1}=\left[\sqrt{\omega}\left(\frac{1}{\sqrt{\omega}}(\chn\left[\rfsb_i\right])^{-1}\frac{1}{\sqrt{\omega}}\right)^\frac{1}{2}\sqrt{\omega}\right]^2,
\end{equation}
where 
\begin{equation}\label{eq:Fherm-algo-eq}
    \rfsb_i=\frac{1}{\sqrt{\chn^\dag[\rfdsb_i]}}\rho\frac{1}{\sqrt{\chn^\dag[\rfdsb_i]}}.
\end{equation}
As $i\rightarrow\infty$, $\rf_i$ will converge to the solution.
\end{algorithm}
The authors postulated the conjecture due to strong numerical evidence, stating that the algorithm always converges to a solution if it exists. When $\alpha$ and $\beta$ are free to be non-Hermitian, the conjecture has already been proven by Friedland \cite{friedland2017schrodinger}. Here, we note that the proof of Theorem~\ref{thm:quantum-surjective} may be used to prove the theorem for Hermitian $\alpha$ and $\beta$ too:
\begin{theorem}\label{thm:gp-conjecture}
    Given a positivity improving CPTP map $\chn$ and two density matrices $\rho$ and $\omega$, there exists invertible matrices $\poi$ and $\pof$ and Hermitian matrices $\rfsb$ and $\rfdsb$ that satisfy Eq.~\eqref{eq:qsb-constraint-1} to \eqref{eq:qsb-constraint-4}. In particular, we can also set $\poi = \sqrt{\chn^\dag[\rfdsb]}$ and $\pof = \sqrt{\rfdsb}$.
\end{theorem}
\begin{proof}
        To avoid encumbering the main text, the proof can be found in Appendix~\ref{app:thm-gp-conjecture}.
\end{proof}
Note that this does not prove the convergence of Algorithm~\ref{algo:qsb}, but simply the existence of a solution. Additionally, we leave the question of uniqueness of the solution open, though some results have been shown in \cite{friedland2017schrodinger}.

With this all detailed, one may expect that $\chndh$ is such that a quantum analogue of Theorem \ref{thm:CSB-is-CPH} can be obtained. This, however, is generally \textit{not} the case. Specifically ${(\hat\chndh)_\rho} \neq \retd$, where $\retd$ is prior-hacked for $(\chn,\rho,\omega)$. This boils down to the non-correspondence between Eqs.~\eqref{eq:QPH-algo-eq} and \eqref{eq:Fherm-algo-eq}, that is setting $\Xi \mapsto \Gamma$ in Eq.~\eqref{eq:qm-SB-channel-standard-hermitian} does not yield a proper QSB. For details, see Appendix \ref{app:Fherm-not-Thm}.

\subsubsection*{Quantum Analogue to Theorem \ref{thm:CSB-is-CPH} \& \\ Inference-Consistent Quantum Schr\"odinger Bridges}
Now, while $\chndh$ does not satisfy a quantum analogue of Theorem \ref{thm:CSB-is-CPH}, this does not mean that no $\chnds$ fulfils this feature. There may be some choice of $\poi$ and $\pof$ for Eq.~\eqref{eq:qm-SB-channel-standard} that does so. This brings us to our key result for our exploration in the quantum regime. 

\begin{theorem}\label{thm:QSB-is-QPH}
    For quantum Schr\"odinger bridge or quantum prior hacking problems characterized by $(\chn,\rho,\omega)$, the quantum Schr\"odinger bridge $\chnd$ satisfies 
    \begin{equation} \label{eq:QSB-is-QPH}
        \ret{\chnd}{\rho} = \retd
    \end{equation}
    (where $\retd[\omega]=\rho$) if and only if $\chnd$ is given by $\chnds$ with 
    \begin{align}\label{eq:condn-ic}
        \poi &= \sqrt{\rho} \ir \rf, \quad \pof = \sqrt{\omega} \ir{\chn[\rf]}
    \end{align}
i.e. \begin{equation} \label{eq:ic-QSB}
    \chnd [\cdot]= \sqrt{\omega} \ir{\chn[\rf]} \chn\bigg[ \sqrt \rf \ir \rho\cdot \ir \rho \sqrt{\rf} \bigg] \ir{\chn[\rf]} \sqrt{\omega}
\end{equation}
\end{theorem}

\begin{proof}
    For the \textit{if} direction. we must check two implications. Firstly, that Eq.~\eqref{eq:condn-ic} ensures that $\chnds$ is indeed a QSB (after this choice of $\poi,\pof$ we label it as $\chnd$). Secondly, it is the case that under such conditions $\ret{\chnd}{\rho} = \retd$. For the first implication, with Eq.~\eqref{eq:condn-ic} 
    \begin{equation}\label{eq:fgen-condned}
       \chnds[\rho] = \sqrt{\omega} \ir{\chn[\rf]} \chn\bigg[ \sqrt \rf \ir \rho \rho \ir \rho \sqrt{\rf} \bigg] \ir{\chn[\rf]} \sqrt{\omega}, 
    \end{equation}
which simply gives $\omega$, and so this is a valid QSB, sending $\rho$ to $\omega$. Now, for the second implication, we see easily that, from Eq.~\eqref{eq:qm-SB-channel-standard}:
\begin{align*}
    \chnds^\dagger [\cdot] &= \poi^{-\dagger} \chn^\dagger[\pof^\dagger \cdot \pof] \poi^{-1}  \\ \Rightarrow
    \chnd^\dagger[\cdot]&= \ir \rho \sqrt{ \rf } \,\chn^\dagger \bigg[ \ir{\chn[\rf]} \sqrt{\omega}\cdot \sqrt\omega \ir{\chn[\rf]} \bigg] \sqrt{ \rf } \ir \rho
\end{align*}
Now, noting Eq.~\eqref{eq:petz-map} and the fact that $\chnd$ is proven as a QSB, this just means,
\begin{align*}
    \ret{\chnd}{\rho}[\cdot] &= \sqrt \rho \ir \rho \sqrt{ \rf } \,\chn^\dagger \bigg[ \ir{\chn[\rf]} \sqrt{\omega} \ir \omega \cdot \\ & \quad  \quad \ir \omega \sqrt\omega \ir{\chn[\rf]} \bigg] \sqrt{ \rf } \ir \rho \sqrt \rho \\
    &= \retd[\cdot]
\end{align*}
So the \textit{if} direction is proven. As for the \textit{only if} direction, we work backwards. First we impose that $\ret{\chnd}{\rho} = \retd$:
\begin{align*}
    & \; \sqrt{\rho}\chnd^\dag\left[\frac{1}{\sqrt{\chnd[\rho]}}\cdot\frac{1}{\sqrt{\chnd[\rho]}}\right]\sqrt{\rho}  \\ = & \;  \sqrt{\rf}\chn^\dag\left[\frac{1}{\sqrt{\chn[\rf]}}\cdot\frac{1}{\sqrt{\chn[\rf]}}\right]\sqrt{\rf}.  
\end{align*}  
We then obtain the adjoint, which gives the channel itself, 
\begin{align}
    \chnd^\dagger[\cdot] &= 
    \frac{1}{\sqrt{\rho}} \sqrt{\rf} \chn^\dag \bigg[\ir{\chn [\rf]} \sqrt{\chnd[\rho]} \cdot \nonumber \\ & \quad \quad \sqrt{\chnd[\rho]}\ir{\chn [\rf]}\bigg] \sqrt{\rf} \frac{1}{\sqrt{\rho}} \nonumber \\
    \chnd[\cdot] &= \sqrt{\chnd[\rho]}\ir{\chn [\rf]}
    \chn
    \bigg[ \sqrt{\rf} \frac{1}{\sqrt{\rho}} 
    \cdot \nonumber \\ & \quad \quad  \frac{1}{\sqrt{\rho}} \sqrt{\rf}  \bigg] \ir{\chn [\rf]} \sqrt{\chnd[\rho]}  \label{eq:modified-chnd-raw}
\end{align}
Finally we apply QSB condition $\chnd[\rho]=\omega$, we obtain Eq.~\eqref{eq:ic-QSB}, and thus Eq.~\eqref{eq:condn-ic}.
\end{proof}

We may call $\chnd$ of Eq.~\eqref{eq:ic-QSB}, the \textit{inference-consistent} quantum Schr\"odinger bridge, due to its uniqueness in satisfying Eq.~\eqref{eq:QSB-is-QPH}.
We make a few notes here. 
\begin{itemize}
    \item When $\chn$ is injective, Eq.~\eqref{eq:modified-chnd-raw} and Eq.~\eqref{eq:QSB-is-QPH} \textit{implies} that $\chnd$ is a QSB for $\rho,\omega$ \footnote{A parallel feature occurs if we perform such a similar induction for a classical SB. The structure of the classical bridge $\chnd$ which goes to Eq.~\eqref{eq:classical-SB-transit} implies $\chnd p =q$ if $\chn$ is an injective stochastic map.}. Consider how $\ret{\chnd}{\rho}[\omega]=\retd[\omega]=\rho$ gives $\chnd^\dag [\chnd[\rho]^{-\frac{1}{2}}\omega\chnd[\rho]^{-\frac{1}{2}}]=\mathbb{I}$. Now, $\chn$ is injective if and only if $\retd$ is injective, which is the same as $\ret{\chnd}{\rho}$, so $\chnd$ is injective and thus $\chnd^\dag$ is injective. Hence, if $\chn$ is injective, $\chnd[\rho]^{-1/2} \, \omega \, \chnd[\rho]^{-1/2} = \mathbb{I}$ and so $\chnd[\rho]=\omega$ \footnote{If one applies the same arguments for inducing an $\chnd$, but now with minimum change principle map as found in \cite{bai2025quantum} with $\omega$ as the output reference, one simply gets the outcome that if $\chnd$ is to be a SB, then $[\chn[\rf],\omega]=0$. This simply means that the recovery maps of \cite{bai2025quantum} reduce to the Petz transpose.}. 

    \item The inference-consistent QSB's $\Xi$ and $\xi$ of \cite{georgiou2015positive}'s formulation is equivalent to $\Gamma$ and $\rf$ respectively, mirroring the classical matrix scaling scenario. 
    
    \item The inference-consistent QSB can be written in the following way (c.f. Eq.~\eqref{eq:ic-QSB}) :
    \begin{equation}
    \chnd = \mathcal{J}_{\sqrt{\omega}\ir {\chn[\rf]}}\circ \chn \circ \mathcal{J}_{\sqrt{\rho}\ir {\rf}}^{-1}
    \end{equation}
    where $\mathcal{J}_{\chi}[\cdot] = \chi \cdot \chi^\dagger$. This is structurally very much comparable to the \textit{classical} SB scenario as per Eq.~\eqref{eq:classical-SB-transit}, which may be written as
    \begin{equation}
            \chnd = D_{q \oslash \chn\rf} \, \chn \, D_{p \oslash \rf}^{-1}.
    \end{equation}

    \item As mentioned before, $\chnds$ is not formulated in terms of an optimisation problem. Given that the Petz map, in its modified form in \cite{bai2025quantum}, is a solution to an optimisation problem, one might wonder whether the inference-consistent QSB $\chnd$ could then be phrased as a solution to an optimisation problem. Unfortunately, a straightforward quantisation of the optimisation in the classical SB problem would give neither $\chnds$ nor $\chnd$ as a solution, leaving the problem open for now. For details, see Appendix~\ref{app:optimisation-sb-quantum}.
\end{itemize}
The uniqueness of the inference-consistent QSB $\chnd$ among generic candidates for QSBs $\chnds$ is notable given that Theorem \ref{thm:CSB-is-CPH} presides over the classical regime. 

\section{Concluding Discussion} \label{sec:concl}
We conclude by summarizing our findings alongside their conceptual upshots, before mentioning possible avenues for future work. 

\subsubsection*{Bayes' Rule, Schr\"odinger Bridges \& Belief Kinematics}
We have showed that pathological closed-mindedness can be spoofed as sound Bayesian updating. Theorem \ref{thm:classical-prior-hacking-channel-forall} demonstrates the broad susceptibility of Bayesian inference to such prior hacking. Remarkably, this vulnerability is categorically parallel to the problem of constructing Schr\"odinger bridges (Theorem \ref{thm:CSB-is-CPH}).

This correspondence has notable implications for belief kinematics and philosophical discussions of epistemic pathologies \cite{Battaly2020-dogma-philo, cassam2016vice-dogma-philo}. We emphasize that in statistical physics, the Schr\"odinger bridge is far from pathological: it represents a principled update of a reference process to the most reasonable neighbouring dynamics consistent with the empirical constraints $p,q$ (Eq.~\ref{eq:optimisation-problem-sb}). However, when interpreted through the lens of belief dynamics, the mathematical structure admits a pathological reading. If prior hacking corresponds to retroactively altering one’s prior $\rf$ to protect a desired belief, then Schr\"odinger bridging corresponds to revising one’s model of the evidence-generating process $\chn$ itself in order to preserve that belief.

In informal terms, while prior hacking resembles the claim: “I forgot to mention that I have hidden evidence that I was never sick,” Schr\"odinger bridging corresponds to: “I realize now that doctors cannot be trusted after all.” Theorem \ref{thm:CSB-is-CPH} shows that these two forms of epistemic inertia are mathematically equivalent. Altering one’s model of evidence propagation after the fact is formally indistinguishable from manipulating the prior so that Bayesian updating yields the desired conclusion. But once again, this pathological reading is only possible if $p$ had \textit{not} been a genuine observation but just some pre-selected belief we do not wish to update on.

Returning to the role of Schr\"odinger bridges in statistical physics and its applications, our link to prior hacking (Theorem \ref{thm:CSB-is-CPH}) demonstrates that these constructions are really doing something characteristically Bayesian in nature. In effect, they implement some dual of Bayes’ rule, not on prior distributions themselves, but on the underlying process, given genuine observations on the input $p$ and output $q$. Given Bayes' rule's epistemic significance, this perspective elucidates the same significance embedded in Schr\"odinger bridges.

\subsubsection*{The Petz Recovery Map \& Quantum Schr\"odinger Bridges}
Our quantum extension for the Petz recovery reveals a similar susceptibility to prior hacking as in the classical setting (Theorem \ref{thm:quantum-surjective}). Theorem \ref{thm:QSB-is-QPH} is particularly noteworthy as it shows that the Petz map obeys a quantum analogue to Theorem \ref{thm:CSB-is-CPH} defined for QSBs as constructed in the existing literature \cite{georgiou2015positive}. Furthermore, the result \textit{selects} a unique inference-consistent candidate for QSBs among the current continuum of generic candidates. 

\subsubsection*{Avenues for Future Work}
First, unlike in the classical setting, quantum Schr\"odinger bridges are not yet known to arise from a well-defined optimization problem. Appendix \ref{app:optimisation-sb-quantum} discusses possible reasons for this gap. Identifying such an optimization principle would be a natural extension of current work, particularly in light of the inference-consistent formulation of QSBs introduced here.

Second, our analysis has focused primarily on existence results. A natural follow-up is to understand how ``easy prior hacking is'' for given tuples $(\chn,p,q)$. Our simulations and figures suggest that highly dissipative process are easier to prior hack. Erasure channels are trivially hackable in both classical and quantum settings while channels with high absolute determinants require more fine-graining.  Studying robustness under constrained prior adjustments may provide a new operational yet belief-kinematic perspective on the irreversibility of a channel $\chn$ with respect to states $p,q$.

Finally, given the strong structural correspondence established here between Bayesian inference and Schr\"odinger bridges, it may be worthwhile to examine whether existing applications of Schr\"odinger bridge methods in statistical physics and machine learning can be helpfully conceptualized through the lens of prior hacking, and whether such a perspective offers any advantages for associated computational tasks.


\section*{Acknowledgements}
We thank Minjeong Song, Valerio Scarani, Jeremy Heng and Tryphon T. Georgiou for helpful discussions. This project is supported by the National Research Foundation, Singapore through the National Quantum Office, hosted in A*STAR, under its Centre for Quantum Technologies Funding Initiative (S24Q2d0009); and by the Ministry of Education, Singapore, under the Tier 2 grant ``Bayesian approach to irreversibility'' (Grant No.~MOE-T2EP50123-0002).


\bibliographystyle{unsrt}
\bibliography{apssamp}

\appendix
\onecolumngrid
\section{Hadamard product}\label{app:hadamard}

The Hadamard product is a binary operation between two $\mathds{R}^{m\times n}$ matrices that is defined by
\begin{equation}
    (A \odot B)_{i,j} = A_{i,j}B_{i,j}.
\end{equation}
I.e., it is an elementwise multiplication. It is commutative, assosciative, and distributive over addition. An important identity is
\begin{equation}\label{eq:hadamard-identity}
    a \odot b = D_ab=D_ba.
\end{equation}
Moreover, we can similarly define a division operation $\oslash$ such that if $A \odot B = C$, $A = C \oslash B$ and $B = C \oslash A$. With these, we can take Eq.~\eqref{eq:classical-prior-hacking-problem} and rewrite it as
\begin{align}
    D_{\rf}\chn^\tp D_{\chn\rf}^{-1}q &=p \nonumber\\
    \rf\odot\left[\chn^\tp (q \oslash (\chn\rf))\right]&=p \nonumber\\
    \rf&=p\oslash[\chn^\tp (q\oslash(\chn\rf))].
\end{align}
Here, we have arrived at Eq.~\eqref{eq:classical-prior-hacking-equation}. Note that the division might not be well defined, which will affect the existence of the solution. See the main text for more information.

\section{The Iterative Proportional Fitting (IPF)/Sinkhorn problem}\label{app:ipf-sinkhorn}
Let us introduce the Iterative Proportional Fitting (IPF) problem, also known as the Sinkhorn problem. Given an initial $d'\times d$ matrix $X$, we wish to find a new matrix 
\begin{equation}
    M=D_aXD_b,
\end{equation}
where $D_a,D_b$ are diagonal matrices built from vectors $a,b$ such that
\begin{align}
    Me&=u \nonumber \\
    M^\tp e&=v,
\end{align}
where $e=(1,1,\cdots,1)^\tp$. Here, $u,v$ are vectors denoting the row and column sum of $M$. If we insert the definition of $M$ into the first condition, we find
\begin{align}
    D_aXD_be&=u \nonumber\\
    a \odot (Xb) &= u.
\end{align}
Similarly, for the second condition,
\begin{equation}
    (X^\tp a) \odot b = v.
\end{equation}
By substituting $a$, we get
\begin{equation}\label{eq:IPF}
    b=v\oslash[X^\tp(u\oslash(Xb))].
\end{equation}
This is the same as Eq.~\eqref{eq:classical-prior-hacking-equation}, by a simple change of variables:
\begin{equation}
    (X,u,v,b)\longleftrightarrow(\chn,p,q,\rf).
\end{equation}
Note that if we define $\rfd=q\oslash\chn(\rf)$, we could also exchange $a$ and $\rfd$. Therefore, the prior hacking problem is mathematically equivalent to the Sinkhorn problem.

\section{Proofs}
\subsection{Proof for Theorem \ref{thm:classical-prior-hacking-channel-forall}}\label{app:thm-classical-hacking-forall}
We first consider the equivalence of the first two statements. (\textit{1.}) holds if and only if there are no singularities emergent from the entries in the $D^{-1}_{\chn \rf}$ term, which is equivalent to the statement that all outputs on $\chn$ have full support:
\begin{equation}\label{eq:fullsupp}
    \forall (s \in \Delta^{d-1},y \in \Omega'): [\chn s](y) >0.
\end{equation}
This is nothing but (\textit{2.}). Firstly, from (\textit{2.}) to Eq.~\eqref{eq:fullsupp} is trivial. Secondly, from Eq.~\eqref{eq:fullsupp} to (\textit{2.}), consider first the negation of (\textit{2.}): $\exists (y,x'):\chn_{y,x'}=\chn(y|x') = 0$. Then for a pure input $s(x)=\delta_{xx'}$, $\chn s (y)=\sum_{x}\chn(y|x) s(x)=\chn(y|x')=0 $, which contradicts (\textit{2.}). Thus, the (\textit{1.}) and (\textit{2.}) are equivalent.

(\textit{2.}) implies (\textit{3.}) because of Theorem \ref{thm:classical-prior-hacking}. The matrix $B_{y,x}=p(x)q(y)$, whose rows sum to any $q$ and columns sum to $p$, will always have the same pattern as $\chn$ because it is a positive matrix, thus satisfying the second condition of Theorem~\ref{thm:classical-prior-hacking}. As for how (\textit{3.}) implies (\textit{2.}), consider the contrapositive which is that if
\begin{equation}
    \exists j,i:\chn_{j,i} = 0,
\end{equation}
then there exists $p,q$ such that prior-hacking is impossible. This can be simply proven by picking $p(x)=\delta_{x,i}$ and $q(x)=\delta_{y,j}$. The only nonnegative matrix that has a row sum of $q$ and column sum of $p$ is the matrix $B_{y,x}=\delta_{x,i}\delta_{y,j}$. However, this matrix violates the third condition of Theorem~\ref{thm:classical-prior-hacking}, which means that prior-hacking is impossible.


\subsection{Corollaries on Prior Hacking Concatenations} \label{app:concat-coros}
Here, we state some corollaries that involve concatenations of channels and prior hacking. The first is a sensible implication of Theorem \ref{thm:classical-prior-hacking-channel-forall}.

\begin{coro}\label{coro:concat-coro1}
    If two transition matrices $\chn_1: \Omega 
    \longrightarrow \Omega''$ and $\chn_2:\Omega'' 
    \longrightarrow \Omega'$ are individually prior-hackable for tuples of conclusions $p$ and evidences $q$ (in respective state spaces) then their concatenation $\chn_2\chn_1:\Omega 
    \longrightarrow \Omega'
    $ will also be prior-hackable for any probability distributions $p \in \Omega$ and $q \in \Omega'$.
\end{coro}
\begin{proof}
    $(\chn_2\chn_1)_{y,x} = \sum_{x} (\chn_2)_{y,z}(\chn_1)_{z,x}>0 $ since $(\chn_2)_{y,z}>0$ and $(\chn_1)_{z,x}>0$. Thus, by Theorem \ref{thm:classical-prior-hacking-channel-forall}, $\chn_2\chn_1$ is prior-hackable for all $p,q$.
\end{proof}

The converse, however, does not hold. Channels that individually may have some $p,q$ that cannot be prior-hacked, can be concatenated together to become prior-hackable for all $p,q$. We can also extend the corollary to make some statements on $\chn^n$. 

\begin{coro}\label{coro:concat-coro2}
    For a transition matrix $\chn:\Omega \longrightarrow \Omega$, $\exists n \in \mathds{R} \; s.t. \;\chn^n$ is prior-hackable for all $p,q \in \Delta^{d-1}$ if and only if $\chn$ is primitive (that is, irreducible and aperiodic).
\end{coro}
\begin{proof}
    By definition, a matrix $\chn$ is primitive or regular if and only if $\exists n \in \mathds{R} \; s.t. \;\chn^n(y|x)>0$ for all $x,y$. On the other hand, from Theorem~\ref{thm:classical-prior-hacking-channel-forall}, $\chn^n$ is prior-hackable for all $p,q$ if and only if $\chn^n(y|x)>0$ for all $x,y$. Thus, the two statements are equivalent \cite{meyer2023matrix, herstein1954note-primitive-matrices}. 
\end{proof}

In particular, this may be equivalently stated $\chn$ are not primitive if and only if for all $n$ (for any length of self-concatenations) there is always has some $p,q$ that one cannot prior hack $\chn^n$ for. Though this is certainly an aside, it does give an explicitly Bayesian definition of matrices that are primitive and not primitive.

\subsection{Proof for Theorem~\ref{thm:quantum-surjective}}\label{app:thm-quantum-surjective}
First, we define the operation
\begin{equation}\label{eq:q-prior-hack-channel}
    \mathcal{X}_{\chn,\omega}(\rf)=\sqrt{\rf}\chn^\dag\left[\frac{1}{\sqrt{\chn[\rf]}}\omega\frac{1}{\sqrt{\chn[\rf]}}\right]\sqrt{\rf}.
\end{equation}
This is the nothing more than the Petz recovery map, but with $\rf$ as the input instead. Since the channel $\chn$ is such that $\chn[\rho]$ is full rank for all $\rho$, the map is always well defined for all choices of $\rf$ because $\text{supp}(\omega) \subseteq \text{supp}(\chn[\rf])$ for all choices of $\rf$. What we want to prove is that for a finite-dimensional system, the map $\mathcal{X}_{\chn,\omega}$ is surjective, which means that the whole space of $\rho$ is in the image space of this map and that there is always at least one $\rf$ that gives $\retd(\omega)=\rho$.

We first illustrate the proof for qubit systems. In Bloch sphere terms, for qubits, this map is defined on the Bloch Ball $\mathsf B^3$. We first have to prove two things: that $\mathcal{X}_{\chn,\omega}$ is continuous, and that it is the identity at the boundary of the Bloch sphere $\partial \mathsf B^3=S^2$. The fact that it is continuous can be readily seen from the fact that all of the operations involve continuous functions. The only place in which a discontinuity could have formed is the $\chn[\rf]^{-\frac{1}{2}}$. However, because $\chn[\rf]$ is full rank, the inverse is guaranteed to exist. Moreover, at the boundary of $\mathsf B^3$, $\rf$ is a pure state, thus,
\begin{align}
    \mathcal{X}_{\chn,\omega}(\ket{\rf}\!\bra{\rf})\bigg\vert_{\partial B^3}&=\ket{\rf}\!\bra{\rf}\chn^\dag\left[\frac{1}{\sqrt{\chn[\rf]}}\omega\frac{1}{\sqrt{\chn[\rf]}}\right]\ket{\rf}\!\bra{\rf}\nonumber\\
    &=\ket{\rf}\!\bra{\rf}.
\end{align}
Thus, at the boundary, $\mathcal{X}$ is an identity on the boundary of the Bloch ball. With these two conditions, we can prove that $\mathcal{X}$ is surjective using a topological construction that relies on the no-retraction theorem and is often used in proving Brouwer’s fixed-point theorem. For more background information, we refer the readers to \cite{milnor1997topology}.

Assume that $\mathcal{X}$ is non-surjective. Thus, there is a state $\alpha\in\text{Interior}\,(\mathsf B^3)$ such that $\alpha \notin \text{Image}(\mathcal{X}_{\chn,\omega})$. We can then define another function $f$ that, for each point in $\mathsf B^3$, projects the point along a line that emanates from $\alpha$ to the point onto the boundary $\partial \mathsf B^3$. Thus, the composed map $f \circ \mathcal{X}_{\chn,\omega}$ will take every point in $\mathsf B^3$ and map it continuously to  $\partial \mathsf B^3$, which is a continuous retraction of the ball into its boundary, which is disallowed by the no-retraction theorem, which states that $S^{n-1}$ is not a retract of $\mathsf B^n$. Thus, $\mathcal{X}$ must be surjective.

In higher-dimensional systems, the state space is it is no longer a ball. However, it is still a convex, compact set with a non-empty interior, and is thus homeomorphic to a ball. Similarly, the boundaries are still formed by the pure states which are homeomorphic to a sphere. Therefore, the same proof still applies to higher-dimensional systems.

\begin{figure}[h!]
    \centering
    \begin{subfigure}{0.49\linewidth}
        \centering
        \includegraphics[width=\linewidth]{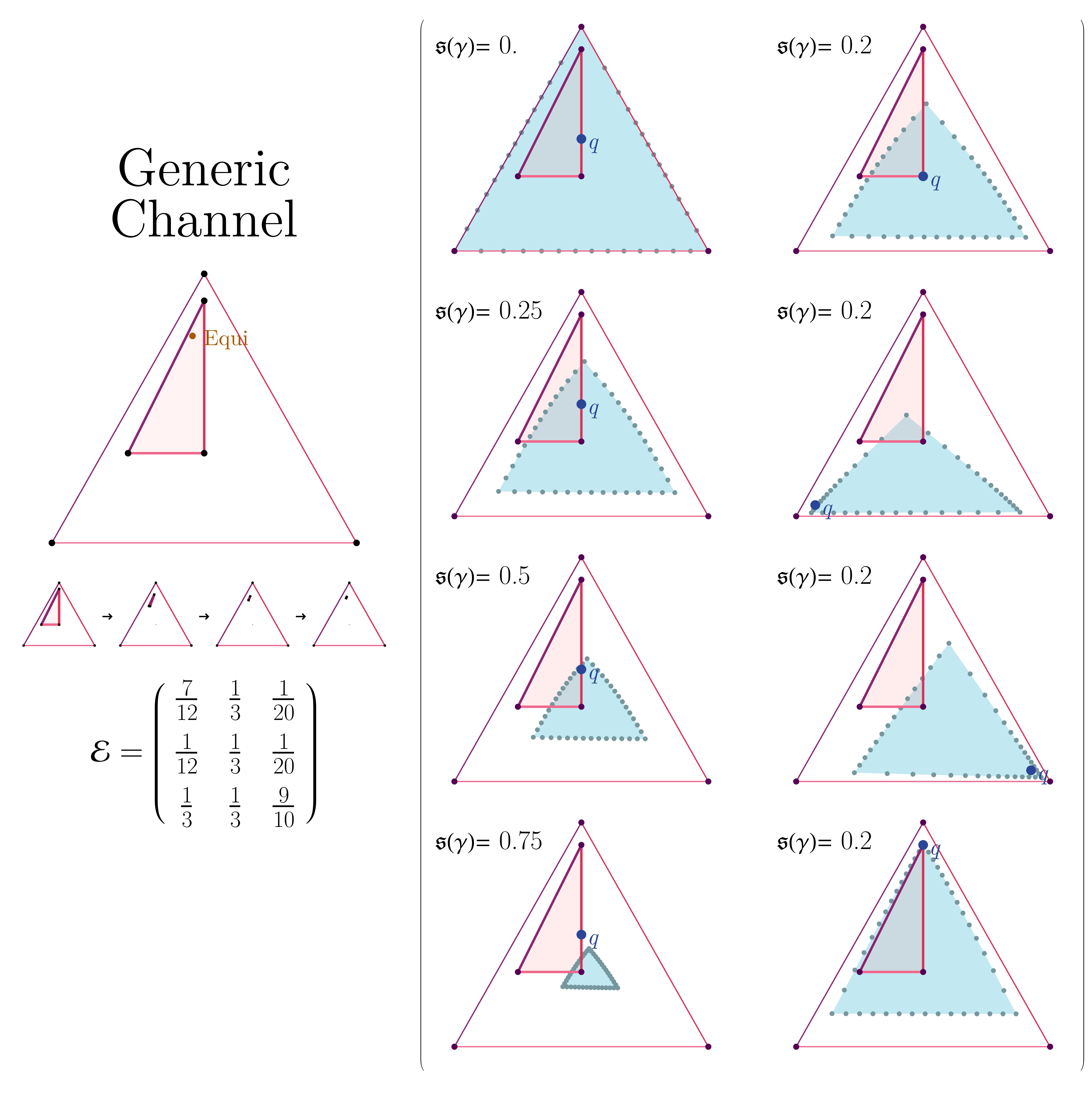}
        \caption{A somewhat generic channel acting on a trit space.}
        \label{fig:cgen}
    \end{subfigure}
    \hfill
    \begin{subfigure}{0.49\linewidth}
        \centering
        \includegraphics[width=\linewidth]{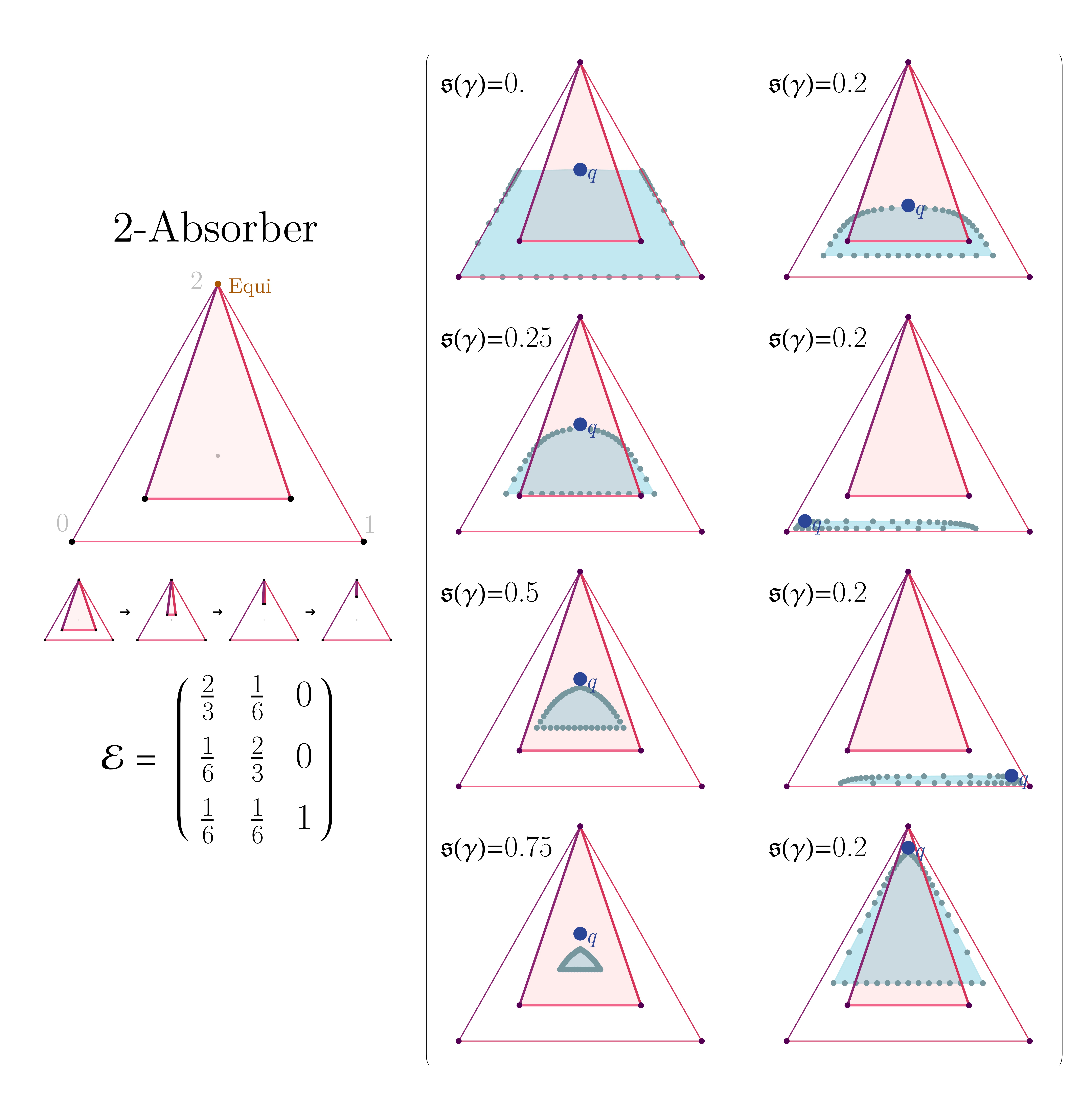}
        \caption{$2$-absorber channel acting on a trit space. Notice how one can only perform a Bayesian update toward states with weights on the $(0,1)$-state that are higher than that of $q$. Thus, prior-hacking is not always possible for any given $p,q$, as per Theorem \ref{thm:classical-prior-hacking-channel-forall}.}
        \label{fig:c2abs}
    \end{subfigure}

    \vspace{0.5cm}

    \begin{subfigure}{0.5\linewidth}
        \centering
        \includegraphics[width=\linewidth]{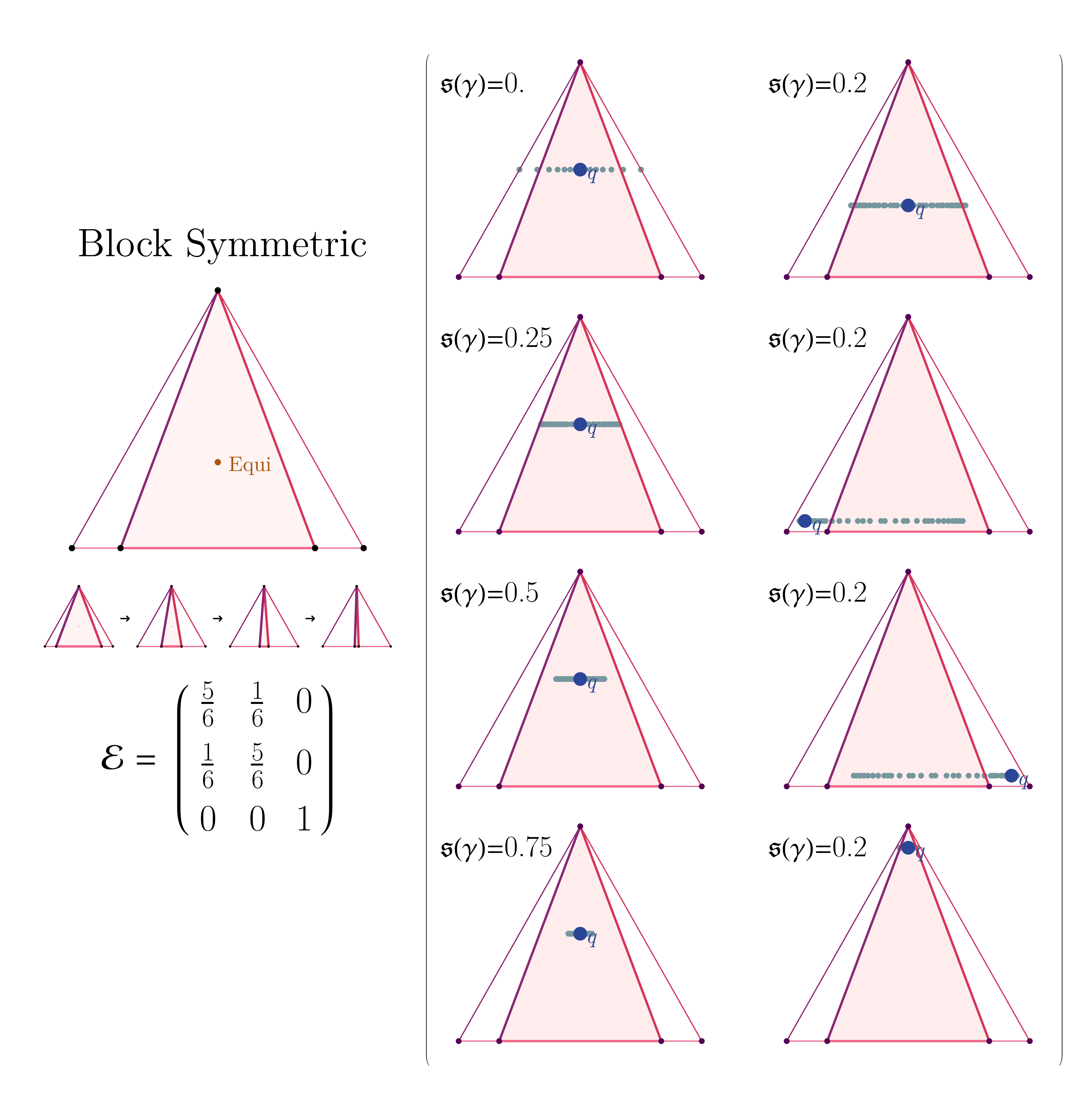}
        \caption{Block-symmetric channel acting on a trit space. Since the channel always holds the probability weight on the $2$-state totally unchanged, the Bayesian update only acts on a line, only updating on a bit-space consistent with the evidence. This illustrates how prior-hacking this channel is not always possible for any given $p,q$, as per Theorem \ref{thm:classical-prior-hacking-channel-forall}.}
        \label{fig:cblocksym}
    \end{subfigure}

    \caption{Image plots of representative channels acting on a trit space. For details, see Section \ref{sec:c-example-text}.}
    \label{fig:c-examples-in-app}
\end{figure}

\subsection{Proof for Theorem~\ref{thm:gp-conjecture}}\label{app:thm-gp-conjecture}

If we look at $\chnds$ from Eq.~\eqref{eq:qm-SB-channel-standard}, we can always do the polar decomposition $\poi = U_0\sqrt{\chn^\dag[\rfdsb]}$ and $\pof = U_T\sqrt{\rfdsb}$ for some unitaries $U_0,U_T$. If $U_0=U_T=\mathds{1}$, then we get back the Hermitian bridge $\chndh$ from Eq.~\eqref{eq:qm-SB-channel-standard-hermitian}. Next, define the operation
\begin{equation}
    \mathcal{Y}_{\chn,\rho}(\rfdsb,U_0,U_T)=U_T\sqrt{\rfdsb}\chn\left[\frac{1}{\sqrt{\chn^\dag[\rfdsb]}}U_0^\dag \rho U_0\frac{1}{\sqrt{\chn^\dag[\rfdsb]}}\right]\sqrt{\rfdsb}U_T^\dag,
\end{equation}
which is simply Eq.~\eqref{eq:qm-SB-channel-standard} but with $\rfdsb,U_0,U_T$ as the inputs. For completeness, we will both prove the theorem for both non-Hermitian (which has already been shown in \cite{friedland2017schrodinger}) and Hermitian $\alpha$ and $\beta$.

First, consider the case that both $\poi$ and $\pof$ are Hermitian, which means $U_0=U_T=\mathds{1}$. The proof of Theorem~\ref{thm:quantum-surjective} (see Section~\ref{app:thm-quantum-surjective})  only requires continuity of $\mathcal{Y}_{\chn,\rho}$ and the fact that $\mathcal{Y}_{\chn,\rho}(\rfdsb)=\rfdsb$ when $\rfdsb$ is pure, both of which are true, given that $\chn$ is positivity improving. Thus, there will always be a $\rfdsb$ such that $\chndh[\rho]=\omega$. In the case that both $\poi$ and $\pof$ are not Hermitian, we can fix $U_0$ and $U_T$ to any values. This does not affect continuity, and the action on the surface is now given by $\chnd[\rho]=U_T\rfdsb U_T^\dag$, which is simply a rigid rotation on the surface of the set, which can be topologically deformed into the identity still. Therefore, by the same logic, using the decomposition $\poi = U_0\sqrt{\chn^\dag[\rfdsb]}$ and $\pof = U_t\sqrt{\rfdsb}$, there will always be a $\rfdsb$ such that $\chnds[\rho]=\omega$, for every fixed value of $U_0,U_T$.

Note that, while we have proven the existence of the solution, the convergence of both Algorithms~\ref{algo:quantum} and \ref{algo:qsb} is still an open question.

\section{Fixed Point Iteration for the Quantum Scenario}\label{app:quantum-ipf-sinkhorn}
Given a channel $\chn$, an evidence $\omega$, and a target state $\rho$, we want to solve the equation
\begin{equation}\label{eq:quantum-prior-hacking}
    \rho = \sqrt{\rf}\chn^\dag\left[\rfd\right]\sqrt{\rf}.
\end{equation}
First, note that for the matrix equation $MNM=L$, one solution for $M$ is
\begin{equation}\label{eq:matrix-sandwich}
M = L^{1/2}(L^{-1/2}N^{-1}L^{-1/2})^{1/2}L^{1/2}.
\end{equation}
Taking the square, we can obtain
\begin{equation}
    \rf=\left[\sqrt{\rho}\left(\frac{1}{\sqrt{\rho}}(\chn^\dag\left[\rfd\right])^{-1}\frac{1}{\sqrt{\rho}}\right)^\frac{1}{2}\sqrt{\rho}\right]^2,
\end{equation}
which we can use as an update rule, similar to the classical scenario. Note that Eq.~\eqref{eq:matrix-sandwich} can be rewritten in a different manner, for example with $L$ in the middle instead of $N$, which will lead to equivalent update rules.


\section{Optimisation problem for the quantum Schrödinger bridge}\label{app:optimisation-sb-quantum}

Unlike the classical scenario, the standard QSB is formulated without starting from a minimisation problem. The question of the distance measure to be minimised to obtain the quantum Schödinger problem was left open in \cite{georgiou2015positive}. On the other hand, the Petz map can be obtained as a solution for the quantum version of the Minimum Change Principle (MCP) \cite{bai2025quantum}. Thus, it is natural to ask whether the optimisation problem for the QSB can be obtained by quantising the MCP.

The quantisation can be done by substituting probability distributions with quantum states, stochastic matrices with CPTP maps, and joint probability distributions with Quantum States Over Time (QSOT). In \cite{bai2025quantum}, the QSOT is defined from a channel and an input state:
\begin{equation}
    \chn \star \rho = \left(\mathds{I}_B \otimes \sqrt{\rho^\tp}\right) C_\chn \left(\mathds{I}_B \otimes \sqrt{\rho^\tp}\right),
\end{equation}
where $C_\chn$ is the Choi–Jamiołkowski isomorphism of the channel $\chn$. The marginals of the QSOT would give the input and output states:
\begin{equation}
    \Tr_B[\chn \star \rho] = \rho, \quad \quad \Tr_A[\chn \star \rho] = \chn[\rho].
\end{equation}
Thus, a straightforward quantisation of Eq.~\eqref{eq:optimisation-problem-sb} is given by
\begin{equation}\label{eq:optimisation-problem-sb-quantum}
\begin{array}{r c l}
    \chnd = & \displaystyle \argmin_{\chnd'} & \displaystyle \textsf{D}(\chnd'\star\rho \| \chn \star \eta). \\ 
        & \text{subject to}                & \displaystyle \Tr_B[\chnd' \star \rho] = \rho \\
        &                            & \displaystyle \Tr_A[\chnd' \star \rho] = \omega
\end{array}
\end{equation}
Here $\eta$ is the input marginal distribution that is in the initial prior process. It is the quantum analogue of $\sum_yA(x,y)$ in the classical version of the problem. However, here we quickly encounter a problem: the solution of this optimisation depends on $\eta$, whereas in the classical problem it does not. In the classical version, we have
\begin{align*}
    \textsf{D}(B\|A) &= \sum_{x,y} B(x,y) \log \left(\frac{B(x,y)}{\chn(y|x)\eta(x)}\right) \nonumber \\
    &= \sum_{x,y} B(x,y) \log B(x,y)-\sum_{x,y}B(x,y) \log\chn(y|x)-\sum_{x}p(x) \log\eta(x).
\end{align*}
Here, the last term is a constant that is fixed by the setup of the problem, thus the solution will not depend on $\eta$. However, the same cannot be said for the quantised version of the problem, regardless of whether one uses the Umegaki relative entropy \cite{umegaki1962conditional}, the Belavkin-Staszewski relative entropy \cite{belavkin1982c}, or the quantum fidelity. Thus, the QSB, either in its generic form $\chnds$ in Eq.~\eqref{eq:qm-SB-channel-standard} or the inference-consistent form $\chnd$ in Eq.~\eqref{eq:ic-QSB} will not be a solution of Eq.~\eqref{eq:optimisation-problem-sb-quantum} as it is currently written, as they both do not depend on $\eta$. With these notes, we leave the problem of formalising an optimisation problem for the Schrödinger bridge open for future exploration.

\section{$\chndh$ does not meet the conditions for Theorem \ref{thm:QSB-is-QPH}}\label{app:Fherm-not-Thm}
Setting $\Xi \mapsto \Gamma$ for Eq.~\eqref{eq:qm-SB-channel-standard-hermitian}, alongside Eq.~\eqref{eq:Fherm-algo-eq}, we get
\begin{align*}
    \mathcal{F}^{\mathsf{\Xi \mapsto\Gamma }}_\mathsf{Herm} [\rho]= \sqrt{\ir{\chn [\rf]}\omega\ir{\chn [\rf]}}\chn\left[\ir{\ir \rf \rho \ir \rf} \rho \ir{\ir \rf \rho \ir \rf}\right]\sqrt{\ir{\chn [\rf]}\omega\ir{\chn [\rf]}} \neq \omega.
\end{align*}

Of course, if it so happens that $[\omega,\chn[\rf]]=0$ and $[\rho,\rf]=0$, then we do have a valid QSB. But this is generally not the case, and so $\chndh$ does not feature inference consistency as per Theorem \ref{thm:QSB-is-QPH}. 

One final remark here is that the symmetry between the Hermitian QSB $\chndh$ and quantum prior hacking may be traced to the classical regime. There, we have the two coupled equations
\begin{align}
    \rfd &= q \oslash (\chn\rf) \\
    \rf &= p \oslash (\chn^\tp\rfd).
\end{align}
When the problem is quantised, the prior hacking approach leads naturally to the relations
\begin{align}
    \rfd &= \mathcal{D}_\omega(\chn[\rf]) \\
    \rf &= \mathcal{D}_\rho'(\chn[\rfd]),
\end{align}
where $\mathcal{D}$ and $\mathcal{D}$ are defined in Eqs~\eqref{eq:quantum-algo-step-1} and \eqref{eq:quantum-algo-step-2}. On the other hand, the Hermitian QSB leads naturally to
\begin{align}
    \rfdsb &= \mathcal{D}_\omega'(\chn[\rfsb]) \\
    \rfsb &= \mathcal{D}_\rho(\chn[\rfdsb]).
\end{align}
Note the slight difference in the choice of the operations. Therefore, these opposing transformations leads to them being symmetrical but inequivalent problems. As previously alluded to, it is clear that if everything commutes (that is, if everything is classical) these two problems do in fact collapse to being the same. 

\section{More Details on Prior hacking for Illustrative Qubit Channels} \label{app:qm-chan}
\subsection{Depolarising channels}\label{app:depolarising-channel}
The depolarising channel is perhaps the most common type of channel in which the image is full rank over the whole domain. It is defined as
\begin{equation}
    \depolarise[\rho_0]=\lambda\frac{\mathds{I}}{2}+(1-\lambda)\rho_0,
\end{equation}
with some parameter $\lambda$ that determines the strength of the channel. Simply put, it is a mixing of the original input state with white noise.

Here we will first summarise the behaviour of the Petz map for the depolarising channel, while the full mathematical expressions can be found in the next sections. First, setting $\lambda=1$, we get $\ret{\depolarise}{\rf}[\omega]\big\rvert_{\lambda=1}=\gamma$, which recovers the erasure channel result of making prior hacking trivial. Conversely, setting $\lambda=0$ yields the unitary channel result with $\ret{\depolarise}{\rf}[\omega]\big\rvert_{\lambda=0} =\vec{\omega}$, i.e., no dependence on $\rf$ and making prior hacking impossible.

Next, we denote the length of the Bloch vector $\rf$ as $|\vec{r}|$, which is simply a measure of the purity of $\rf$. If we evaluate it for a fully mixed state prior ($|\vec{r}|=0$), we get $\ret{\depolarise}{\rf}[\omega]\big\rvert_{|\vec{r}|=0} = (1-\lambda)\omega + \lambda\frac{\mathds{1}}{2}=\omega'$. On the other hand, for a pure state prior ($|\vec{r}|=1$) we have $\ret{\depolarise}{\rf}[\omega]\big\rvert_{|\vec{r}|=1} = \rf$. For a fixed purity value $|\vec{r}|$ in between the two extremes, the set of output states would form a deformed sphere inside the Bloch sphere that is shifted in the $\hat{z}$ direction. In other words, as we shift from $|\vec{r}|=0$ to $|\vec{r}|=1$, this deformed sphere will continuously go from covering just a single point $\omega'$ to the whole Bloch sphere. Because $\ret{\depolarise}{\rf}[\omega]$ is continuous everywhere, there cannot be any `tearing', `puncturing', or sudden `jumps' in this transition. This means that every point in the Bloch sphere must be hit at some point by some value of $\rf$, thus making prior hacking possible. This is the intuition behind Theorem~\ref{thm:quantum-surjective}.

\subsubsection{Deriving the Petz map}
The channel is defined by the operation
\begin{equation}
    \depolarise(\rho)=\lambda\frac{\mathds{I}}{2}+(1-\lambda)\rho,
\end{equation}
with $0\leq\lambda\leq1$, with the extreme cases being a unitary and an erasure channel. The Kraus operators are given by
\begin{align*}
    K_0&=\sqrt{1-\frac{3}{4}\lambda}\mathds{I}, \nonumber\\
    K_i&=\sqrt{\frac{\lambda}{4}}\sigma_i, \qquad i \in \{x,y,z\}.
\end{align*}
Let $\rf$ be a state with the Bloch vector $\vec{r}=(r_x,r_y,r_z)$, with length $|\vec{r}|$. Let $\hat{n}=\frac{\vec{r}}{|\vec{r}|}$. The eigenvectors of $\rf$ are then given by $\ket{\hat{+n}}$ and $\ket{\hat{-n}}$, with the associated eigenvalues $\mu_\pm=\frac{1\pm|\vec{r|}}{2}$. Moreover, the depolarising channel does not change the eigenbasis of the state, thus the eigenvectors of $\depolarise[\rf]$ remain the same, while its associated eigenvalues, denoted by $a_i$, are given by
\begin{equation}
    a_\pm=\frac{\lambda}{2}+(1-\lambda)\mu_\pm=\frac{1\pm|\vec{r}|(1-\lambda)}{2}.
\end{equation}
Finally, define $P_\pm = \ket{\pm\hat{n}}\bra{\pm\hat{n}}$.

The depolarising map is self-adjoint, and thus the Kraus operators of the Petz recovery map are given by
\begin{equation}
    M_\alpha=\rf^{\frac{1}{2}}K_\alpha\depolarise[\rf]^{-\frac{1}{2}}.
\end{equation}
For the identity term, we have $M_0 = \sqrt{1-\frac{3}{4}\lambda}\sum_{j=\{+,-\}}\sqrt{\frac{\mu_j}{a_j}}P_j$. Meanwhile, for the other terms, we have $M_i = \sqrt{\frac{\lambda}{4}}\sum_{j,k=\{+,-\}}\sqrt{\frac{\mu_j}{a_k}}P_j\sigma_iP_k$. Applying the Kraus operators to a state $\omega$, we get
\begin{equation}
    M_0\omega M_0^\dag=\left(1-\frac{3}{4}\lambda\right)\sum_{j,k}\sqrt{\frac{\mu_j\mu_k}{a_ja_k}}P_j\omega P_k,
\end{equation}
for the first operator, and
\begin{align}
    \sum_{i=1}^3M_i\omega M_i^\dag& = \frac{\lambda}{4}\rf^{\frac{1}{2}}\sigma_i\depolarise[\rf]^{-\frac{1}{2}} \omega \depolarise[\rf]^{-\frac{1}{2}}\sigma_i\rf^{\frac{1}{2}} \nonumber \\
    &= \frac{\lambda}{4}\rf^{\frac{1}{2}}\left[2\Tr\left(\depolarise[\rf]^{-\frac{1}{2}}\omega\depolarise[\rf]^{-\frac{1}{2}}\right) -\depolarise[\rf]^{-\frac{1}{2}}\omega\depolarise[\rf]^{-\frac{1}{2}}\right]\rf^{\frac{1}{2}},
\end{align}
for the other operators. The second equality is obtained by applying the Pauli matrix identity $\sum_{i=0}^3\sigma_i A \sigma_i=2\Tr(A)\mathds{1}-A$. The second term is identical to the $M_0\omega M_0^\dag$ term except for the $\lambda$ coefficient, which allows them to be summed up together to $1-\lambda$. The first term can be written as
\begin{align}
    \frac{\lambda}{2}\rf\Tr(\omega\depolarise[\rf]^{-1})=\frac{\lambda}{2}\rf\Tr\left[\frac{\omega P_+}{a_+}+\frac{\omega P_-}{a_-}\right].
\end{align}
Finally, since the channel is spherically symmetric, without loss of generality, we can fix $\omega$ to the Bloch vector $\vec{s}=(0,0,s)$. This will allow us to simplify the first term further into
\begin{align*}
    \frac{\lambda}{2}\rf\Tr\left[\frac{\omega P_+}{a_+}+\frac{\omega P_-}{a_-}\right]&=\frac{\lambda}{4}\rf\Tr\left[\frac{1+n_zs}{a_+}+\frac{1-n_zs}{a_-}\right] \nonumber \\
    &= \lambda\rf\left[\frac{1-|\vec{r}|n_zs(1-\lambda)}{1-(1-\lambda)^2|\vec{r}|^2}\right].
\end{align*}
Finally, the map is then given by
\begin{align}
    \retd(\omega)=\frac{1-\lambda}{2}\sum_{j,k}\sqrt{\frac{\mu_j\mu_k}{a_ja_k}}P_j(\mathds{1}+s\sigma_z) P_k + \lambda\rf\left[\frac{1-|\vec{r}|n_zs(1-\lambda)}{1-(1-\lambda)^2|\vec{r}|^2}\right].
\end{align}

Let us expand the first term. The terms that are diagonal in the eigenbasis of $\rf$, given by $j=k$ give
\begin{equation}
    \frac{\mu_\pm}{a_\pm}P_\pm(\mathds{1}+s\sigma_z)P_\pm=\frac{\mu_\pm}{a_\pm}(1\pm sn_z)P_\pm,
\end{equation}
while the off-diagonal terms with $j\neq k$ would yield $\sqrt{\frac{\mu_+\mu_-}{a_+a_-}}P_\pm(\mathds{1}+s\sigma_z)P_\mp=s\sqrt{\frac{\mu_+\mu_-}{a_+a_-}}P_\pm \sigma_zP_\mp$. We then apply the identity $P_+ \sigma_zP_- + P_- \sigma_zP_+ = \sigma_z-n_z(\hat{n}\cdot \vec{\sigma})$ to simplify it to
\begin{equation}
    s\sqrt{\frac{\mu_+\mu_-}{a_+a_-}}(\sigma_z-n_z(\hat{n}\cdot \vec{\sigma})).
\end{equation}

The Bloch vector of the final state is given by $\Tr(\vec{\sigma}\retd(\omega))$. The $\rf$-diagonal terms give $\frac{1-\lambda}{2}\left[\left(\frac{\mu_+}{a_+}-\frac{\mu_-}{a_-}\right)+sn_z\left(\frac{\mu_+}{a_+}+\frac{\mu_-}{a_-}\right)\right]\hat{n}$, the off-diagonal terms give $(1-\lambda)s\sqrt{\frac{\mu_+\mu_-}{a_+a_-}}(\hat{z}-n_z\hat{n})$, while the final term give $\lambda\left[\frac{1-|\vec{r}|n_zs(1-\lambda)}{1-(1-\lambda)^2|\vec{r}|^2}\right]\vec{r}$. Denote the output vector to be $\vec{b}_{\depolarise,\rf,\omega}$. Combining the three terms, we have
\begin{align}
    \vec{b}_{\depolarise,\rf,\omega}&=
    \frac{1-\lambda}{2}\left[\left(\frac{\mu_+}{a_+}-\frac{\mu_-}{a_-}\right)+sn_z\left(\frac{\mu_+}{a_+}+\frac{\mu_-}{a_-}\right)\right]\hat{n} + 
    (1-\lambda)s\sqrt{\frac{\mu_+\mu_-}{a_+a_-}}(\hat{z}-n_z\hat{n}) + 
    \lambda\left[\frac{1-|\vec{r}|n_zs(1-\lambda)}{1-(1-\lambda)^2|\vec{r}|^2}\right]\vec{r} \nonumber \\
    &=
    \frac{1-\lambda}{2}\left[\left(\frac{\mu_+}{a_+}-\frac{\mu_-}{a_-}\right)+sn_z\left(\sqrt{\frac{\mu_+}{a_+}}-\sqrt{\frac{\mu_-}{a_-}}\right)^2\right]\hat{n} + 
    (1-\lambda)s\sqrt{\frac{\mu_+\mu_-}{a_+a_-}}\hat{z} + 
    \lambda\left[\frac{1-|\vec{r}|n_zs(1-\lambda)}{1-(1-\lambda)^2|\vec{r}|^2}\right]\vec{r}.
\end{align}

\subsubsection{Finding the limits}
Now we will derive the limits of the map as outlined. First, setting $\lambda=1$ recovers the erasure channel result with
\begin{equation}
    \vec{b}_{\depolarise,\rf,\omega} \bigg\rvert_{\lambda=1} = \vec{r},
\end{equation}
meaning that to obtain any evidence $\rho$ we just need to set $\rf=\rho$. Meanwhile, the other extreme of $\lambda=0$ would completely erase the dependence on $\rf$ :
\begin{equation}
    \vec{b}_{\depolarise,\rf,\omega} \bigg\rvert_{\lambda=0} = s\hat{z}=\vec{s},
\end{equation}
which makes prior hacking impossible. Next, when we evaluate it at $|\vec{r}|=0$, we get $\mu_\pm = a_\pm = \frac{1}{2}$ and the first and third term go to zero, leaving only
\begin{equation}
    \vec{b}_{\depolarise,\rf,\omega} \bigg\rvert_{|\vec{r}|=0} = (1-\lambda)s \hat{z}.
\end{equation}
This is simply the Bloch vector associated with the state $\omega'=(1-\lambda)\omega + \lambda\frac{\mathds{1}}{2}$. On the other hand, when we evaluate $|\vec{r}|=1$, we get
\begin{equation}
    \vec{b}_{\depolarise,\rf,\omega} \bigg\rvert_{|\vec{r}|=1} = \vec{r}.
\end{equation}

\subsection{Dephasing channels}

The dephasing channel $\dephase$ is given by Kraus operators
\begin{equation}
    Z_0 = \begin{pmatrix} 1 & 0 \\ 0 & \sqrt{1-\lambda} \end{pmatrix}, \quad Z_1 = \begin{pmatrix} 0 & 0 \\ 0 & \sqrt{\lambda} \end{pmatrix}.
\end{equation}
Intuitively, it reduces the off-diagonal terms in the input state, bringing it closer to a decoherent state. It is an example of a channel that does not fulfil the requirement of Theorem~\ref{thm:quantum-surjective}. For $\rho_0=\ket{i}\bra{i}$, $\dephase[\ket{i}\bra{i}]=\ket{i}\bra{i}$ and is thus not full rank. Because of this reason, it is quite straightforward to see that prior hacking is impossible if the target conclusion is either one of the pure computational state $\rho=\ket{i}\bra{i}$. Normally a pure choice conclusion is prior-hackable by setting $\rf$ to be the same pure state. However, in the case for this channel, the Petz map is undefined for $\rf=\ket{i}\bra{i}$ due to its rank deficiency. However, there are other examples of a combination of $(\rho,\omega)$ that is prior-hackable still. This is most easily seen in the \textit{completely dephasing channel}, given by taking $\lambda=1$.

\begin{theorem}\label{thm:dephasing-channel}
    For a completely dephasing channel $\dephase$, an evidence $\omega$, and a decoherent conclusion $\rho$, prior hacking is possible if and only if both $\omega$ and $\rho$ have the same diagonals in the $Z$-basis, i.e., $\rho=\dephase[\omega]$.
\end{theorem}
\begin{proof}
   The prior hacking equation to be solved is given by
    \begin{equation}
        \rho = \sqrt{\rf}\dephase\left[\frac{1}{\sqrt{\dephase[\rf]}}\omega\frac{1}{\sqrt{\dephase[\rf]}}\right]\sqrt{\rf}.
    \end{equation}
    Note that the dephasing channel is self-adjoint. We can further simplify it by:
    \begin{align}\label{eq:dephasing-prior-hacking}
        \rho &= \sqrt{\rf}\dephase\left[\frac{1}{\sqrt{\dephase[\rf]}}\omega\frac{1}{\sqrt{\dephase[\rf]}}\right]\sqrt{\rf} \nonumber \\
        &= \sqrt{\rf}\dephase[\rf])^{-1/2}\dephase[\omega](\dephase[\rf])^{-1/2}\sqrt{\rf} \nonumber \\
        &= \sqrt{\rf}(\dephase[\rf])^{-1}\dephase[\omega]\sqrt{\rf}.
    \end{align}
    Now, define
    \begin{equation}
        D = (\dephase[\rf])^{-1}\dephase[\omega] := \begin{pmatrix} d_0 & 0 \\ 0 & d_1 \end{pmatrix},
    \end{equation}
    with $d_0,d_1\geq0$. Moreover, they cannot be both zero, as that requires $\omega$ to have zero trace. Now, write
    \begin{equation}
        \sqrt{\rf}:=\begin{pmatrix} x & y \\ y^* & z \end{pmatrix},
    \end{equation}
    with $x,z\geq0$ and $xz\geq|y|^2$. Then, we can get
    \begin{equation}
        \rho = \sqrt{\rf}D\sqrt{\rf} = \begin{pmatrix} x^2 d_0 + |y|^2 d_1 & y(x d_0 + z d_1) \\ y^*(x d_0 + z d_1) & |y|^2 d_0 + z^2 d_1 \end{pmatrix}.
    \end{equation}
    If $\rf$ has any initial coherence, then $|y|>0$, which also implies $x,z>0$. However, since $\rho$ is decoherent, this means that $y(x d_0 + z d_1)=0$. Since it is impossible for both $d_0$ and $d_1$ to be zero, this presents a contradiction. Thus, $\rf$ must be decoherent too. However, if $\rf$ is decoherent, then $\dephase[\omega]$ commutes with $\rf$ and Eq.~\eqref{eq:dephasing-prior-hacking} can be written as
    \begin{equation}
        \rho = \dephase[\omega].
    \end{equation}
    Thus, prior hacking is possible if and only if both $\rho$ and $\omega$ are in the same $Z$-coordinate in the Bloch sphere.
\end{proof}



\end{document}